\documentclass{article}
\usepackage{csquotes}
\usepackage[english]{babel}

\usepackage[letterpaper,top=2cm,bottom=2cm,left=3cm,right=3cm,marginparwidth=1.75cm]{geometry}
\usepackage[
  backend=biber,
  style=alphabetic,   
  sorting=nyt
]{biblatex}

\usepackage{mdframed}
\usepackage{amsmath}
\usepackage{amssymb}
\usepackage{amsthm}
\usepackage{graphicx}
\usepackage{braket}
\usepackage{comment}
\usepackage[colorlinks=true, allcolors=blue]{hyperref}
\theoremstyle{definition} 
\newtheorem{theorem}{Theorem}[section] 
\newtheorem{definition}[theorem]{Definition}
\newtheorem{fact}[theorem]{Fact} 
\newtheorem{lemma}[theorem]{Lemma} 

\usepackage{caption}
\usepackage{enumitem}
\usepackage{thm-restate}
\usepackage{multirow}

\newcommand{\Tr}{\operatorname{Tr}}

\newcommand{\eps}{\varepsilon}
\newcommand{\tr}{\mathrm{tr}}
\newcommand{\state}[1]{\ket{#1}\bra{#1}}
\newcommand{\F}{\mathsf{F}}

\title{Near-optimal entanglement-communication tradeoffs for remote state preparation}
\author{Srijita Kundu \\
\small Quantum Computing Research Centre \\
\small Hon Hai (Foxconn) Research Institute \\
\small \texttt{srijita.kundu@foxconn.com.sg} \\
\and Olivier Lalonde \\
\small Institute for Quantum Computing \\
\small University of Waterloo \\
\small \texttt{olalonde@uwaterloo.ca}
}
\date{}

\begin{document}
\maketitle
\begin{abstract}
    We study the following task: Alice is given a classical description of a rank-$k$ projector $P$ on $\mathbb{C}^d$, and Alice and Bob want to prepare the quantum state $P/k$ on Bob's side using shared entanglement and classical communication. The general form of this task is known as remote state preparation (RSP). We give nearly-matching (up to small additive factors) lower and upper bounds for the entanglement cost and communication cost for RSP of the states $P/k$. Ours are the first nearly-matching upper and lower bounds for RSP of mixed states, and in the special case of pure states, our lower bound outperforms the best previously known lower bound. Our results show that any pure entangled state that can be used to do RSP of these states with $o(d)$ bits of communication, can distill $\log d$ ebits of entanglement, and conversely, any state that can distill $\log d$ ebits of entanglement can be used to do RSP of these states efficiently. As applications of our results, we rederive a previously-known incompressibility result for states of the form $P/k$, and give a new entanglement-assisted communication protocol for the equality function that uses $\frac{1}{2}\log n + O(1)$ many ebits, and  $O(1)$ communication. 
\end{abstract}

\section{Introduction}

A central goal of quantum information theory is to understand what resources are required to transmit quantum states between a sender and a receiver. There are various settings in which such a transmission can be performed. One setting is the task of \emph{quantum teleportation}, in which the sender Alice holds a quantum state that is unknown to her and wants to transmit this to the receiver Bob. Alice and Bob pre-share an entangled quantum state which can be used as a resource for this task, and moreover, Alice can communicate with Bob using classical bits. One of the foundational results of quantum information theory is that a $d$-dimensional quantum state can be teleported using a $d$-dimensional maximally entangled state (i.e., $\log d$ many EPR pairs or ebits) shared between Alice and Bob, and $2\log d$ bits of classical communication \cite{BBC+93}, and that this is the optimal cost.

Moreover, it is known that the tasks of quantum teleportation and \emph{entanglement distillation} are essentially equivalent. In entanglement distillation, Alice and Bob share some arbitrary (potentially mixed) entangled state, and want to distill a maximally entangled state of some dimension from this state, via local operations and classical communication (LOCC). It is easy to show that if some mixed entangled state can be used to teleport $d$-dimensional quantum states with some communication, then Alice and Bob can distill a $d$-dimensional maximally entangled state from this shared state, using the same communication as the teleportation protocol.\footnote{The construction here is to simply run the teleportation protocol in superposition over a basis.}

A task quite similar to the teleportation task is \emph{remote state preparation} (RSP). Here too Alice wants to send a quantum state to Bob using shared entanglement and classical communication, but the difference is that Alice knows the description of the quantum state she wants to send. Due to this difference, this task can be done more efficiently \cite{Bennett_2001}: using only $\log d + \log\log d + O(1)$ communication and the same $d$-dimensional maximally entangled state, for arbitrary pure states. The tradeoff between classical communication and entanglement required for remote state preparation of pure states has been further studied in many works \cite{DB01,Lo00,bennett2005remote}. \cite{bennett2005remote} showed lower bounds for RSP of pure states: they showed that a pure entangled state that can be used for RSP of $d$-dimensional pure states with $o(d)$ bits of classical communication cannot have Schmidt rank smaller than $d$. Moreover, they showed that $\log d - O(1)$ bits of classical communication are needed regardless of the amount of entanglement.

RSP for mixed states has often been studied under the name one-shot quantum state compression \cite{JRS03, jain2005, JSR08, Jain06, BNR18, Bab_Hadiashar_2020}.\footnote{Technically, one-shot state compression is a more general setting where Alice can also send quantum states to Bob, and she wants to send fewer qubits than the number in the target state. But due to teleportation, it is sufficient to consider only shared entanglement and classical communication.} One-shot state compression has connections to several other tasks of interest in quantum information theory, such as \emph{state splitting} \cite{Dev06, BCR11, AJ22}, and \emph{state merging} \cite{HOW05, HOW06, AJ22}, as well as communication complexity \cite{JRS03, jain2005, AJM+16}. \cite{JRS03, BNR18} gave a protocol for state compression for an arbitrary ensemble of mixed states, which is efficient in communication, but uses a very large amount of entanglement. They also showed a communication lower bound of $\log d - O(1)$ for a certain restricted class of protocols for an ensemble of $d$-dimensional mixed states. \cite{Bab_Hadiashar_2020} generalized this lower bound for general protocols, and in fact showed a communication-entanglement tradeoff: they showed that the sum of the communication cost and entanglement cost (measured in number of shared EPR pairs) for a particular ensemble of $d$-dimensional mixed states should be at least $\log d - O(1)$. \cite{AJ22} gave a protocol with optimal communication and a much smaller amount of entanglement than previous protocols, but their entanglement cost was still far from $\log d$.

Unlike for teleportation, RSP with mixed states as the shared entanglement resource has not really been studied. In all the works mentioned above, it is assumed that Alice and Bob share EPR pairs, and the entanglement cost is measured in terms of these. A variant of RSP known as oblivious RSP was studied in \cite{LS03}, and it was shown that this is essentially equivalent to teleportation in the exact (zero-error) case. The requirement in oblivious RSP is that Bob prepares exactly one copy of the target state and does not get any further information about the state. In the exact case, oblivious RSP requires $2\log d$ communication like teleportation, and it is not difficult to show that like teleportation, if oblivious RSP can be achieved using some shared mixed entangled state, EPR pairs can be distilled from the same shared state using the same communication. However, for RSP without the obliviousness condition, entanglement cost for mixed states, or even entanglement measures other than Schmidt rank for pure states, has not been studied at all. In particular it is not understood whether there is some connection between the ability to do RSP with little communication with an entangled state\footnote{Note that we can do RSP without any entanglement at all if allowed $\Omega(d)$ bits of communication.} and the ability to distill EPR pairs from that entangled state.

\subsection{Our results}
In this work, we take a closer look at RSP with arbitrary pure states as the shared entanglement. We specifically study RSP protocols for \emph{flat states}.

\paragraph{Flat states.} A $d$-dimensional rank-$k$ flat state is a mixed state of the form $\frac{P}{k}$, where $P$ is a rank-$k$ projector on $\mathbb{C}^d$ (they are so called due to their spectrum being flat at $\frac{1}{k}$). We will call the set of all such states $G(d,k)$, and a protocol for remote state preparing these states will be called a $(d,k)$-RSP protocol.

Flat states are an important class of states in the context of RSP: they were the ensemble studied in \cite{JRS03, Bab_Hadiashar_2020} as well. Flat states are in a sense complete for RSP: RSP for arbitrary states can be done using RSP for flat states via a method called the Brothers Extension \cite{AJM+16} (the method was proposed in \cite{AJM+16} but only called the Brothers Extension in later works). Moreover, it can be seen that in any one-way communication protocol between Alice and Bob which involves shared EPR pairs and projective measurements, the state on Bob's side at the end of the protocol must be a flat state. Therefore, any protocol of this form is essentially doing RSP of a flat state, and it may be possible to use our results to prove bounds for such protocols. 

\subsubsection{Lower bound}
Our main result shows lower bounds on the entanglement and communication cost of $(d,k)$-RSP protocols. The entanglement measure we use in our lower bound is smoothed entanglement min-entropy, denoted by $H^\eps_{\min}(A)_\sigma$ or $H^\eps_{\min}(B)_\rho$ when Alice and Bob share an entangled state $\ket{\sigma}^{AB}$. This quantity is a one-shot version of the more well-known measure entanglement entropy $H(A)_\rho$ (or $H(B)_\rho$). Although this is only a useful entanglement measure for pure states, it is a stronger measure than the previously studied Schmidt rank, and is related to distillable entanglement for pure states. Moreover, it is possible to extend our result to lower bound the entanglement of formation for mixed entangled states used in RSP; we leave as an open question whether there is a connection to distillable entanglement as well in the mixed state case.

Our resource lower bounds are captured by the following theorem.
\begin{theorem}[Combined version of Theorems \ref{thm:comm-lowerbound} and \ref{thm:minentropy}]\label{thm:lowerbound}
    For all $\gamma > 0$, any $(d,k)$-RSP protocol with relaxed average error $\eps_r$, $m$ bits of communication, and initial shared entangled state $\ket{\sigma}^{AB}$, must satisfy
    \[H^{\delta + \gamma}_{\text{min}}(A)_\sigma \geq \log d - 3 \log(1/\gamma) - O(1) \]
    where $\delta = F\left(\frac{k}{d}+O(\sqrt{\frac{m}{d}}), 1-\varepsilon_r\right) $, and $F$ is a truncated version of the fidelity function.\footnote{See Section \ref{sec:prelim} for a formal definition of $F$. What is important to know is that $F$ is small if the first argument is much smaller than the second.} Moreover, for any (possibly mixed) shared state, the communication $m$ must satisfy
    \[ m \geq \left\lfloor\log\frac{d}{k}\right\rfloor + \log(1-\eps_r).\]
\end{theorem}
The relaxed error measure considered in Theorem \ref{thm:lowerbound} is among the weakest measures of correctness for RSP protocols (weaker than the average error considered in many other works), and our lower bound holds for protocols satisfying even this weak condition.

Note that in order for $H_{\min}^{\delta+\gamma}(A)_\sigma$ to be meaningful, the smoothing parameter $\delta+\gamma$ has to be smaller than $1$. This means that in order to get a nontrivial lower bound on entanglement for fixed $k$, we must have $m \ll d$, which is what we expect. Moreover, for such an $m$, the function $F$ is only smaller than $1$ if $k/d < 1- \eps_r$, i.e., $\eps_r < 1-\frac{k}{d}$. This is also what we expect, because there is a trivial protocol achieving error $1-k/d$, in which there is no entanglement or communication and Bob always outputs the $d$-dimensional maximally mixed state.

Theorem \ref{thm:lowerbound} is an improvement over the result of \cite{Bab_Hadiashar_2020} in several ways. Firstly, the entanglement measure we use is stronger; in fact, one open question in their work was to find a lower bound on entanglement cost stronger than Schmidt rank. Secondly, for protocols with average error $\eps/2$, the communication cost + entanglement cost lower bound in their result is $\log d - 3\log(1/(1-\eps)) - O(1)$, so they do not get any nontrivial bound on entanglement at all for protocols with $\log d$ communication. On the other hand, our lower bound is nontrivial all the way up to $m$ only being smaller than $d$ by a constant factor (depending on the value of $k$). \cite{Bab_Hadiashar_2020} do show an additional entanglement lower bound of $\log d - \frac{1}{2}\log k - O(1)$ for protocols with optimal communication, but our lower bound also outperforms this.
In the special case where $k=1$ and we are considering pure states, our entanglement lower bound is an improvement over the result of \cite{bennett2005remote} due to our stronger measure.

Our result shows a connection between communication-efficient RSP and entanglement distillation, due to the fact that $H^\eps_{\min}(A)_\rho$ captures distillable entanglement for pure states. It is known due to \cite{WTB17} that there exists a one-way entanglement distillation protocol between Alice and Bob which distills $H^{\eps}_{\min}(A)_\sigma - 2\log(1/\eps) - O(1)$ many EPR pairs from their shared pure entangled state $\ket{\sigma}^{AB}$. Moreover, this rate of entanglement distillation is asymptotically optimal --- it was shown in \cite{BBPS96} that the number of EPR pairs that can be distilled from $\ket{\sigma}^{AB}$  in the asymptotic i.i.d. setting is $H(A)_\sigma$ (which $H^\eps_{\min}(A)_\sigma$ converges to). Previous results do not show such a connection, because it is not possible to distill EPR pairs up to the Schmidt rank of a pure entangled state.

\subsubsection{Upper bounds}
We also give efficient protocols for $(d,k)$-RSP that nearly match our communication and entanglement lower bounds. We give two different protocols: our first protocol is based on a previous protocol of \cite{bennett2005remote} for RSP pure states, and has near-optimal entanglement use but suboptimal communication. It is simple to analyze, and moreover, it only has one-sided error in the average case (this means the protocol either fails, or it succeeds and Bob outputs exactly the target state --- the failure probability is upper bounded on average). Our second protocol has the same entanglement cost but lower communication cost, although its error is no longer one-sided. This protocol is original to this work, and is more complicated to analyze than the first protocol. 
\begin{restatable}{theorem}{upperbound}
\label{thm:rsp-upper}
    For all $d \geq k$ and for all $\varepsilon > 0$, there exist $(d,k)$-RSP protocols, with the following parameters:
    \begin{enumerate}
    \item A protocol which uses a maximally entangled state of local dimension $d$, has communication $m =\log \frac{d}{k} + \log \log d + 2\log \frac{1}{\varepsilon_a} + O(1)$, and has one-sided error $\eps_a$ on average;
    \label{item:prot-1}
    \item A protocol which uses a maximally entangled state of local dimension $d$, has communication $m=\log\frac{d}{k} + 6\log\frac{1}{\eps} + 2\log\log\frac{1}{\eps} + O(1)$, and has worst-case error $\eps$. \label{item:prot-2}
    \end{enumerate}
\end{restatable}
The first protocol in Theorem \ref{thm:rsp-upper} has near-optimal entanglement cost, as per the lower bound in Theorem \ref{thm:lowerbound}, but its communication has an additional factor of $\log\log d$; the second protocol avoids this factor. It is possible to give a worst-case correct version of the first protocol, and in that case it no longer has one-sided error, and therefore has no advantages over the second protocol.

It is worth noting that the entanglement and communication cost of the protocols above nearly match the lower bounds in Theorem \ref{thm:lowerbound}, even though the lower bounds work for protocols which are only correct on average. In fact, the way we design the second protocol is by designing a protocol with average-case correctness first (because this is easier). Then we convert this to a protocol that is worst-case correct via the following average-case to worst-case reduction, which works for generic $(d,k)$-RSP protocols, and may be of independent interest.
\begin{restatable}{theorem}{avtoworst}
\label{thm:worstcase}
Suppose there exists a $(d,k)$-RSP protocol with average-case error $\eps_a$, communication $m$, and which uses a shared entangled state $\sigma^{AB}$. Then, for all $\delta > 0$, there exists a $(d,k)$-RSP protocol that achieves worst-case error $\eps_a + \delta$, using the same entangled state $\sigma^{AB}$ as $\mathcal{P}$, and which communicates $m+4\log(1/\delta)+\log\log(1/\delta)+O(1)$ bits.
\end{restatable}
The average-case to worst-case reduction may also be applied to the first protocol, but as stated before, the resulting protocol has no advantages over the second protocol.

We stress that this equivalence between average-case error and worst-case error for $(d,k)$-protocols holds even in the absence of shared randomness. This may be surprising to some considering that, in communication complexity, in the absence of shared randomness, average-case and worst-case error typically give rise to very different complexity measures.


We give a comparison between our protocols and other protocols for $(d,k)$-RSP in the literature in Table \ref{tab:comparison}, where it can be seen that our protocols are significantly better in terms of entanglement cost compared to previous protocols, and comparable in terms of communication cost. Note that most previous works don't consider $(d,k)$-RSP explicitly: we derive the parameters by using their protocols for state compression, state splitting, etc. to derive $(d,k)$-RSP protocols. In addition to these, we can also compare our protocol in the special case $k=1$ to the RSP protocol for pure states given in \cite{bennett2005remote}. Our first protocol in this case is essentially the same as theirs: their communication is $\log d + \log\log d + O(\log(1/\eps))$, although they get a better constant in the $O$ than us.
\begin{table}[!ht]
\centering
\begin{tabular}{|c|c|c|c|}
\hline
\textbf{Protocol} & \textbf{Entanglement} & \textbf{Communication} & Error \\
\hline
Rejection sampling & \multirow{2}{*}{$O\left(\frac{d}{k}\cdot\log\frac{d}{\eps}\right)$} & \multirow{2}{*}{$\log\frac{d}{k} + O(\log\log\frac{1}{\eps})$} & \multirow{2}{*}{Worst-case} \\
\cite{JRS03, BNR18} & & & \\
\hline
Efficient decoupling & \multirow{2}{*}{$O\left(\frac{1}{\eps^2}\cdot\log\frac{d}{\eps}\right)$} & \multirow{2}{*}{$\log \frac{d}{k} + O\left(\log\frac{1}{\eps}\right)$} & \multirow{2}{*}{Average-case} \\
\cite{AJ22} & & & \\
\hline
$\eps$-net protocol & \multirow{2}{*}{None} & \multirow{2}{*}{$\Theta(d\log\frac{1}{\eps})$} & Relaxed \\
\cite{Bennett_2001} & & & average-case \\
\hline
Kraus operator protocol & \multirow{2}{*}{$\lceil\log d \rceil$} & \multirow{2}{*}{$\log\frac{d}{k} + \log\log d + O\left(\log\frac{1}{\eps}\right)$}  & Average-case, \\
Theorem \ref{thm:rsp-prot-1} & & & one-sided \\
\hline
Damped rejection sampling & \multirow{2}{*}{$\lceil\log d\rceil$}  & \multirow{2}{*}{$\log\frac{d}{k} + O\left(\log\frac{1}{\eps}\right)$} & \multirow{2}{*}{Worst-case} \\
Theorem \ref{thm:rsp-prot-2} & & & \\
\hline
\end{tabular}
\caption{Comparison of entanglement cost (in ebits) and communication cost between previously-known protocols, and our protocols, for $(d,k)$-RSP}
\label{tab:comparison}
\end{table}

We also note that our upper bound and lower bound together show that for pure entangled states, usefulness for communication-efficient $(d,k)$-RSP, and ability to distill $\log d - O(1)$ EPR pairs, are one and the same. The lower bound shows that any pure state that is usable for doing $(d,k)$-RSP with $o(d)$ communication has nearly $\log d$ ebits of distillable entanglement. Our first protocol then shows that any pure state from which $\log d$ ebits of entanglement can be distilled is usable for efficient RSP, since one can distill EPR pairs first, and then use our first protocol.\footnote{This does additionally require that the entanglement distillation protocol is efficient in communication, but it can be checked that this is the case for the protocol in \cite{WTB17}.}

\subsubsection{Applications}
\paragraph{Visible compression.} First, as an application of our entanglement lower bound for $(d,k)$-RSP protocols, we rederive a result of \cite{Bab_Hadiashar_2020} on the impossibility of visible compression for an ensemble of flat states. This is the main result of \cite{Bab_Hadiashar_2020}, which supersedes a number of lower bounds in the literature for similar tasks; we can prove an arguably stronger version of the \cite{Bab_Hadiashar_2020} incompressibility result, and therefore supersede these previous results as well.
\begin{restatable}{theorem}{incompress}
\label{cor:incompress}
For any $\eta > 0$, there exists $d_0$ and $C > 0$ such that for all $(k,d,\eps)$ with $d \geq d_0$ and $\eps < 1-
\frac{k}{d}-\eta$, any visible $(d',\eps)$-compression scheme for the ensemble of all flat states in $G(d,k)$ must satisfy 
\[\log d' \geq \log d - C.\]
This means that this ensemble cannot be compressed by more than a constant number of qubits.
\end{restatable}
Visible $(d',\eps)$-compression schemes are formally defined in Section \ref{sec:incompress}, but it essentially means a compression scheme where the compressor has a full description of the state to be compressed, rather than only having the ability to apply a CPTP map to a copy of the state. Our Theorem \ref{cor:incompress} is a bit stronger than the result of \cite{Bab_Hadiashar_2020} because we show incompressibility for any nontrivial error tolerance, whereas \cite{Bab_Hadiashar_2020} only gave this result for $\eps \leq \frac{1}{2}$. The only downside of our result compared to theirs is that they prove the existence of a small finite ensemble of flat states which is incompressible, whereas our ensemble is essentially the whole of $G(d,k)$ (which can be discretized into an extremely large but finite ensemble via an $\eps$-net). We note however that the proof technique of our entanglement lower bound Theorem \ref{thm:minentropy} (see Lemma \ref{prop:majorizes}) can in fact yield the existence of an ensemble that isn't too large (of size polynomial in $d$, let's say) and which is still incompressible.

\paragraph{Communication complexity of equality.} Additionally, as an application of our efficient $(d,k)$-RSP protocols, we give an entanglement-optimal bounded error protocol for the well-studied equality function on $n$ bits. In $\mathrm{EQ}_n$, Alice and Bob are given inputs $x, y \in \{0,1\}^n$ and want to communicate to determine whether these strings are equal or not. In the entanglement-assisted quantum communication setting, Alice and Bob pre-share entanglement and want to communicate classically in order to do this.

It is known that equality can be solved with $O(1)$ communication with either shared randomness or shared entanglement. However, the most well-known entanglement-assisted protocol for this uses $\log n$ ebits of shared entanglement (and these ebits are just used to generate shared randomness and then implement the best classical protocol).

We show that it is possible to halve the number of EPR pairs by using them in a genuinely quantum way. The $n$-dependence in the result is optimal, as it is known that $\frac{1}{2} \log n + O(1)$ qubits of communication are both necessary and sufficient to compute the equality function with constant error probability \cite{lalonde2023tightboundsrandomizedquantum}. Our result is a reproof of their upper bound, using our efficient RSP protocol.
\begin{restatable}{theorem}{eq}
\label{thm:equpper}
For all $\varepsilon > 0$, there exists an entanglement-assisted protocol for $\mathrm{EQ}_n$ with worst-case error probability $\varepsilon$ which uses $\frac{1}{2}\log n + O(\log\frac{1}{\varepsilon})$ shared EPR pairs and $O(\log\frac{1}{\varepsilon})$ classical communication.
\end{restatable}

It should be noted that our result is slightly worse than the upper bound of \cite{lalonde2023tightboundsrandomizedquantum} in that the $\varepsilon$ dependence we get is a bit worse. Moreover, their protocol has perfect completeness, i.e. the protocol always declares Alice and Bob's inputs to be equal if they are, which is not true of our protocol. The point here is to show that this result is a simple corollary of our RSP protocols.

\subsection{Our techniques}
\subsubsection{Lower bound}
Our communication lower bound is easy to derive: it just follows from the fact that flat states in $G(d,k)$ can be used to send $\log\frac{d}{k}$ bits of information. We'll focus our attention on the proof of entanglement lower bound. Due to the fact that the entanglement min-entropy measure has not been considered before in the literature for proving entanglement lower bounds for RSP, our techniques here are quite different from existing techniques in the literature.

In order to derive a lower bound on $H_{\min}^\eps(A)_\sigma$, we need to control the Schmidt spectrum of the initial shared state $\ket{\sigma}^{AB}$. The correctness condition of the protocol only tells us about the state of one of Bob's registers at the end of the protocol, so we need to connect the spectrum of the initial state $\sigma$ of the protocol to the spectrum of ensemble of final states $(p(c|P), \chi_{P,c})$ for each possible message $c$ Alice can send to Bob. Here $p(c|P)$ is the probability of the message $c$ when Alice's input is $P$, and $\chi_{P,c}$ is the final state of the protocol for this input and message. We can connect the spectrum of $\sigma$ to the spectrum of this ensemble via the well-known result in \cite{NielsenVidal2001}, which says that since the transformation is done via LOCC, the Schmidt spectrum of $\sigma$ must be majorized by the ensemble average of the Schmidt spectra of $\chi_{P,c}$.

Moreover, we need to make use of the fact that the marginals of $\chi_{P,c}$ on Bob's target register are on average (over $c$ and $P$) close to $\frac{P}{k}$. We'll do this by employing the trick used to distill entanglement from a teleportation protocol: we'll run the protocol in superposition over Alice's input $P$, though this superposition will be over all of $G(d,k)$ instead of just a basis.\footnote{For technical reasons, we actually need to do a finite version of this infinite superposition.} This will let us compare the spectrum of $\sigma$ to the spectrum of the ensemble $(p(c), \rho_c)$, where $\rho_c$ is the expected value of $\chi_{P,c}$ with $P$ drawn from the Haar measure in $G(d,k)$ conditioned on the message being $c$. This is because the reduced state on Bob's side for message $c$ actually is $\rho_c$ in the superposition protocol. These states are much easier to handle than the $\chi_{P,c}$ states for arbitrary $P$.

We notice that if the averaging over $P$ in the definition of $\rho_c$ had been according to the Haar measure instead of the Haar measure conditioned on $c$, then the spectrum of $\rho_c$ would straightforwardly have been like the maximally mixed state $\frac{I}{d}$ (which has $\log d$ min-entropy). This is because averaging $\frac{P}{k}$ over the Haar measure gives $\frac{I}{d}$, and correctness requires that Bob's register in $\chi_{P,c}$ is close to $\frac{P}{k}$.\footnote{Relaxed average-case correctness requires something a bit weaker than this, but we'll ignore this difference for now.} Now the actual averaging in $\rho_c$ is conditioned on a particular $c$, but this can't make things too bad. This is because $c$ is $m$-bits long, and conditioning on a particular $c$ changes probabilities by a multiplicative factor of at most $2^m$. This is not too bad when $m$ is not too large. Using some tricks, it is possible to upper bound the spectrum of $\rho_c$ (for most $c$) by what the spectrum would have been with Haar-averaging, along with an additive factor of $O(\sqrt\frac{m}{d})$. This in turn means that the spectrum of $\rho_c$ is close to the spectrum of the maximally mixed state, with the closeness or smoothing term having a dependence on $O(\sqrt{\frac{m}{d}})$. This lets us lower bound the $H^{\gamma+\delta}_{\min}$ of $\rho_c$ on average over $c$, which lets us get the lower bound for $\sigma$.

\subsubsection{Upper bound}
Our first protocol is a generalization of the protocol of \cite{bennett2005remote} to flat states. It is not too difficult to do this generalization, and we get the same entanglement and communication cost as \cite{bennett2005remote} for pure states for this reason. We focus on the second protocol, which we call damped rejection sampling, which is more novel. We will first describe the average-case version of our second protocol, and then describe the average-case to worst-case reduction.

\paragraph{Damped rejection sampling, average-case.} This protocol can be described as a more refined version of the rejection sampling protocol that has appeared many times in the literature \cite{JRS03, BNR18, Bab_Hadiashar_2020}. The idea behind rejection sampling is the following: suppose Alice and Bob share $N$ copies of the maximally entangled state $\ket{\Phi_+}$ on $\mathbb{C}^d\times\mathbb{C}^d$. If Alice, on input $P$, measures in the basis $\{\overline{P}, I-\overline{P}\}$ on her half of each of these copies, and gets the outcome corresponding to $\overline{P}$ on one copy, then due to the nature of the maximally entangled state, Bob's residual state on his half of the corresponding copy will be $\frac{P}{k}$. Of course, Bob does not know which copy Alice succeeded on, so Alice has to send the index $i\in[N]$ in order for him to output that register. The probability of Alice not obtaining outcome $\overline{P}$ in any of the copies, and therefore this procedure failing, is $(\Tr((I-\overline{P})\state{\Phi_+}))^N = \left(1-\frac{k}{d}\right)^N$, due to $P$ being a rank-$k$ matrix. Taking $N=\frac{d}{k}\ln\frac{1}{\eps}$ makes this probability at most $\eps$, and gives the communication and entanglement cost of this protocol that we've mentioned previously.

We notice that the most important thing in this analysis was the fact that the probability of Alice not succeeding in each iteration was $1-\frac{k}{d}$. This is true if Alice measures a fixed projector $\overline{P}$ on her half of a maximally entangled state, but it's also true if Alice measures a uniformly random rank-$k$ projector on her half of any state. In particular, if she kept on measuring in the basis $\{\overline{P}_i, I - \overline{P}_i\}$, for random projectors $P_1, \ldots, P_N$, sequentially on the same state, her probability of not obtaining the $\overline{P}_i$ outcome for any $i\in[N]$ will be $\left(1-\frac{k}{d}\right)^N$. But what use is measuring random projectors $P_1, \ldots, P_N$ when we have a fixed projector $P$ as our input? As it turns out, for any fixed $P$ and a random unitary $U_i$, $Q_i=U_i \overline{P}U_i^\dagger$ is a random projector. So if Alice and Bob use shared randomness to sample random unitaries $U_1, \ldots, U_N$,\footnote{This sharing of randomness will not be needed in the final protocol, because we will be able to use the probabilistic method to fix a good choice of $U_1, \ldots, U_N$.} and Alice does the measurements $\{Q_i, I-Q_i\}$ sequentially on her half of a shared $\ket{\Phi_+}\in \mathbb{C}^d\otimes\mathbb{C}^d$, her probability of not obtaining the good outcome for any $i\in[N]$ will be $\left(1-\frac{k}{d}\right)^N$.\footnote{In fact, we don't even need to use Haar-random $U_i$-s in order for this to hold. Taking $U_i$-s from a $1$-design suffices, since the only property required is that $\mathbb{E}[Q_i]=\frac{k}{d}I$.}

Unfortunately, even if Alice obtains the good outcome for some $i\in [N]$, the state on Bob's side after this would not be $\frac{U_iPU_i^\dagger}{k}$, so that Bob is able to undo $U_i$ and recover $\frac{P}{k}$. Although this holds in the first iteration when Alice and Bob share $\ket{\Phi_+}$, on reaching the $i$-th iteration, the shared state between Alice and Bob is conditioned on the previous $i-1$ measurements failing. The shared state may still be close to $\ket{\Phi_+}$ if $\frac{k}{d}$ is small and $i$ is not too large (by the Gentle Measurement Lemma), in which case Bob's final state would still be close to $\frac{P}{k}$. But in general we don't have control over how large $\frac{k}{d}$ is. To mitigate this situation, we need to modify or ``damp'' our measurements. Instead of doing the projective measurement $\{Q_i,I-Q_i\}$, we will do the POVM with Kraus operators given by $\{\sqrt{\eta Q_i}, \sqrt{I-\eta Q_i}\}$, for a small damping parameter $\eta$ (we'll need to take $\eta$ to be $\eps^2$ in the end, where $\eps$ is the desired average error). If the first outcome of this measurement is obtained, the state on Alice's side is still exactly $\frac{Q_i}{k}$. The probability of obtaining the first outcome is reduced to $\frac{\eta k}{d}$, so the probability that the first outcome is not obtained after $i$ iterations is $\left(1-\frac{\eta k}{d}\right)^i$.

It may seem like we have worsened our situation by requiring a larger number of iterations (so that we can make the probability of not obtaining the first outcome in any iteration as small as we want), but now we can control the damage to the state in each unsuccessful iteration a lot better. Using arguments similar to the Gentle Measurement Lemma, we show that after each iteration in which we obtain the second outcome, the fidelity between the state conditioned on the second outcome and $\ket{\Phi_+}$ decreases by only $\frac{\eta^2k}{d}$, in expectation over the choice of $U_i$. This means that after $(i-1)$ failed iterations, the fidelity between the shared state at the start of the $i$-th iteration and $\ket{\Phi_+}$ goes like $1-(i-1)\frac{\eta^2k}{d}$. Consequently, if this iteration succeeds, and then Bob undoes $U_i$, the fidelity between the resulting state on Bob's side and $\frac{P}{k}$ goes like $1-(i-1)\frac{\eta^2k}{d}$. Since the probability of getting to the $i$-th iteration at all goes like $\left(1-\frac{\eta k}{d}\right)^{i-1}$, we can lower bound the resulting average fidelity (over the choice of $U_i$, and the $i$ that the protocol may succeed at) by taking $N'=\frac{d}{\eta k}\ln(1/\eps)$ many iterations.

We note that the intuition behind this protocol is similar to Elitzur-Vaidman bomb testing \cite{EV93}, in which $n$ weak measurements are done on a state, each measurement disturbing it by $\sim\frac{1}{n^2}$, so that we are able to get a successful outcome while only disturbing the state by $\sim\frac{1}{n}$ overall.

\paragraph{Average-case to worst-case reduction.} Our average-case to worst-case reduction will use a similar idea of picking some good unitaries $U_1, \ldots, U_N$ probabilistically. Suppose we are given a protocol $\mathcal{P}$ that has average-case error $\eps$, and we want to construct a protocol $\mathcal{P}'$ that has worst-case error $\eps+\delta$. Suppose on input $P$, Alice could look at the minimum error of $\mathcal{P}$ on inputs $\frac{UPU^\dagger}{k}$ for all unitaries $U$, perform $\mathcal{P}$ on input $\frac{UPU^\dagger}{k}$ with Bob, and then have Bob undo $U$. The correctness of this hypothetical protocol will be the best-case correctness of $\mathcal{P}$, which is certainly at least the average-case correctness. But of course Alice cannot find the best $U$ for every $P$ and communicate it to Bob efficiently, so in the actual protocol $\mathcal{P}'$, we will approximate the minimum over the whole unitary group $U(d)$ with a minimum over a finite set of unitaries $U_1, \ldots, U_N$. These $U_1, \ldots U_N$ can be probabilistically selected by sampling i.i.d. from $U(d)$. The set of unitaries for which $UPU^\dagger$ has error not much more than $\eps$ is decent enough, due to the average-case correctness of $\mathcal{P}$. So we can show that the probability that all the unitaries $U_1, \ldots, U_N$ fall in the large-error bad set for a particular $P$ is small, by concentration of measure on $U(d)$. Taking a union bound over a $\delta/4$-net on $G(d,k)$ then gives us worst-case error $\eps+\delta$ for $\mathcal{P}$.

\subsection{Open problems}
\subsubsection{Extension of lower bound to mixed states}
 The most obvious open problem left by our work is to non-trivially extend our entanglement lower bound to mixed entangled states. There is a panoply of entanglement measures for mixed state entanglement. One such measure is entanglement of formation, and our result actually can be extended to lower bound the entanglement of formation for mixed states. Entanglement of formation of a mixed state $\sigma^{AB}$ is defined as
 \[ E_F(\sigma^{AB}) = \inf\left\{\sum_ip_iH(A)_{\psi_i}: \sigma^{AB} = \sum_ip_i\state{\psi_i}^{AB}\right\},\]
 where the infimum is taken over all convex combinations of pure states that equal $\sigma^{AB}$.
 Since we can asymptotically lower bound entanglement entropy of pure states usable for RSP, we can asymptotically lower bound the entanglement of formation by lower bounding the entanglement entropy of each pure state in any convex combination (see e.g. \cite{AY18} where such a strategy is employed to test entanglement of formation of mixed states).
 
 A more intriguing question is whether there is a connection between usability for efficient RSP and entanglement distillability for mixed states. Note that unlike pure states, where every entangled state has some amount of distillable entanglement, not all mixed entangled states have distillable entanglement. So for example, there exist states that have non-zero entanglement of formation but zero distillable entanglement. Moreover, there exist mixed entangled states that are useful in various ways, e.g. for violating Bell inequalities, which do not have distillable entanglement \cite{VB14}.
 
 The entanglement distillation result of \cite{WTB17} also works for mixed states --- their protocol can distill $H^{\eps}_{\min}(A|E)_\sigma - 2\log(1/\eps) - O(1)$ many EPR pairs from a mixed entangled state $\sigma^{AB}$, where $H^\eps_{\min}(A|E)_\sigma$ is a conditional smoothed min-entropy and $\ket{\sigma}^{ABE}$ is a purification of $\sigma^{AB}$. If our result could be extended to a lower bound on $H^\eps_{\min}(A|E)_\sigma$ for mixed entanglement, it would show that entangled states that are useful for doing RSP efficiently do have distillable entanglement, which would be a nontrivial result about mixed state entanglement.

 \subsubsection{Generalization of upper bounds}

 Our second protocol for $(d,k)$-RSP is near-optimal in entanglement and communication cost, with its entanglement cost being significantly better than all previously-known protocols. However, the previous protocols such as rejection sampling \cite{jain2005} and \cite{AJ22} work for more than just RSP. In fact, \cite{ADJ17} do a variety of quantum communication tasks such as state splitting, state redistribution etc using a coherent version of rejection sampling, which uses the same extravagant amount of entanglement. \cite{AJ22} reduces the entanglement cost of these tasks compared to \cite{ADJ17}, but their entanglement cost still has a multiplicative factor of $\frac{1}{\eps^2}$ as shown in Table \ref{tab:comparison}. It would be interesting to see if there is a coherent version of our protocol that has the same entanglement cost. This would give protocols that are near-optimal in communication and entanglement cost for a variety of quantum communication tasks.
 
 Another intriguing question is whether our protocols can be made computationally efficient. \cite{AJ22} give protocols that are worse than ours in terms of communication and entanglement cost, but their protocols are computationally efficient in some sense, essentially due to considering a computationally efficient version of the Convex Split Lemma from \cite{ADJ17}. The average-case version of our second protocol uses efficient unitaries (essentially unitaries from the Pauli group), but the number of iterations is very large (even bigger than that in rejection sampling). Moreover, our average-case to worst-case reduction uses Haar-random unitaries. It would be interesting to see if these issues can be handled, so that the protocol is efficient overall.



\section{Preliminaries}\label{sec:prelim}
In this section, we define quantities we will need for the rest of the paper, and state some technical results, starting with the set of rank-$k$ $d$-dimensional projectors. We assume the reader is familiar with basic concepts in quantum information.
\paragraph{Grassmannian.} The set of all $k$-dimensional linear subspaces of a $d$-dimensional vector space (typically $\mathbb{C}^d$) is called the Grassmannian $G(d,k)$. In our applications, we will identify a $k$-dimensional linear subspace with the projector onto it. So an element of $G(d,k)$ will be a rank-$k$ projector on $\mathbb{C}^d$. We will call a mixed state of the form $\frac{P}{k}$ a flat state, and with some abuse of notation, we will also identify $P\in G(d,k)$ with this flat state. We make the Grassmannian into a metric space by defining the distance between two projectors $P$ and $Q$ to be the trace distance $\|\frac{P}{k} - \frac{Q}{k}\|_\tr$ (trace distance is defined formally later in this section). We will also talk about the Haar measure (or the uniform distribution) on $G(d,k)$. By this we mean sampling a projector $P \in G(d,k)$ in the following way: sampling a $U$ in the unitary group $U(d)$ first, and then taking $P = U_kU_k^\dagger$, where $U_k$ refers to the first $k$ columns of $U$.

All logarithms throughout the paper will be assumed to be in base 2.

\subsection{Quantum information background}
A generalized measurement (POVM) $\{M_i\}_{i=1}^m$ on $\mathbb{C}^d$ is a collection of $d \times d$ complex matrices (sometimes called Kraus operators) such that
\[\sum_i M_i^\dagger M_i = I\]
Given that the pure state $\rho$ is measured according to this measurement, the probability of obtaining outcome $i$ is given by $p_i = \tr(\rho M_i^\dagger M_i)$ and the residual state is $\frac{M_i \rho M_i^\dagger}{p_i}$.

For a given $d \geq 1$, the standard maximally entangled state on $\mathbb{C}^d \otimes \mathbb{C}^d$ is defined as
\[\ket{\Phi_+}^{AB} = \frac{1}{\sqrt{d}} \sum_{i=1}^d \ket{i}^A \ket{i}^B\]
We have the following standard fact:
\begin{fact} \label{fact:maximallyentangled}
    If the first half of $\ket{\Phi^+}$ is measured according to the measurement $\{M_i\}_{i=1}^m$, yielding outcome $i$, the residual state on the first half is given by $\frac{M_i^T M_i}{p_i}$
\end{fact}

\begin{definition}[Schatten $p$-norm]
The Schatten $p$-norm of an operator $A$ is given by $\Vert A\Vert_p = \left[\Tr\left((A^\dagger A)^{p/2}\right)\right]^{1/p}$. If $\{s_i\}_i$ are the singular values (eigenvalues for normal operators) of an operator, then $\Vert A\Vert_p = \left(\sum_i s_i^p\right)^{1/p}$.

Of particular interest to us will be the Schatten $1$-norm (also known as trace norm), the $2$-norm (also known as Frobenius norm) and the $\infty$-norm (also known as spectral norm). The spectral norm is equal to the largest singular value of an operator.
\end{definition}

\begin{definition}[Trace distance]
The trace distance between two quantum states with density matrices $\rho$ and $\sigma$ is given by
\[ \Vert \rho - \sigma\Vert_{\mathrm{tr}} = \frac{1}{2}\Vert\rho-\sigma\Vert_1.\]
\end{definition}

\begin{fact}[Variational characterization of the trace distance] \label{fact:variational}
Given two states $\rho$ and $\sigma$, we have
\[\|\rho-\sigma\|_{\tr} = \sup_{P \text { a projector}} \Tr(P(\rho-\sigma)).\]
\end{fact}

\begin{definition}[Fidelity]
The fidelity between two quantum states $\rho$ and $\sigma$ is given by
\[ \mathsf{F}(\rho, \sigma) = \left[\Tr\left(\sqrt{\sqrt{\rho}\sigma\sqrt{\rho}}\right)\right]^2.\]
If $\rho$ and $\sigma$ are diagonal in the same basis, i.e., they are essentially classical probability distributions $p = \{p_i\}_i$ and $q = \{q_i\}$ with the same support, then the expression for fidelity reduces to
\[ \mathsf{F}(p, q) = \left(\sum_i \sqrt{p_iq_i}\right)^2.\]
\end{definition}
The following collects well-known properties of the fidelity and the trace distance:
\begin{fact} \label{fac:distanceproperties}
Given two quantum states $\rho$ and $\sigma$ on the same system, we have:
\begin{enumerate}
\item For any two pure states $\psi$ and $\phi$, it holds that
\[\mathsf{F}(\ket{\psi}\bra{\psi}, \ket{\phi}\bra{\phi}) = \left|\braket{\psi|\phi}\right|^2.\] \label{item:pure-f}
\item (Unitary invariance) For any unitary $U$, 
\[\|U\rho U^\dagger-U\sigma U^\dagger\|_{\tr} = \|\rho -\sigma \|_{\tr}\]
\[\mathsf{F}(U\rho U^\dagger, U\sigma U^\dagger) = \mathsf{F}(\rho,\sigma).\] \label{item:unitary-inv}
\item (Fuchs-van de Graaf inequality) $\|\rho - \sigma\|_\tr \leq \sqrt{1-\mathsf{F}(\rho,\sigma)}$. \label{item:fvdg}
\item (Monotonicity under partial tracing) If $\rho$ and $\sigma$ are states on the system $AB$, we have 
\[\|\Tr_A(\rho) - \Tr_A(\sigma)\|_\tr \leq \|\rho - \sigma\|_\tr\]
\[\mathsf{F}(\Tr_A(\rho), \Tr_A(\sigma)) \geq \mathsf{F}(\rho, \sigma).\] \label{item:partial}
\end{enumerate}
\end{fact}

We will also need a function representing a truncated version of the classical fidelity on boolean variables.
\begin{definition}[Truncated fidelity function]
For $(x, y) \in [0,\infty)\times [0,1]$, the truncated fidelity function $F(x,y)$ is defined as follows
\[ F(x,y) = \begin{cases} \left(\sqrt{xy} + \sqrt{(1-x)(1-y)}\right)^2 & \text{ if } x \leq y \\ 1 & \text{otherwise.} \end{cases} \]

\end{definition}

$F(x, y)$ is essentially fidelity between the binary distributions $\{x,1-x\}$ and $\{y,1-y\}$, but we need to set it to $1$ for $x > y$ (and also allow $x$ to be bigger than $1$), because of technical reasons related to our application. The following property of $F$ can be seen.
\begin{lemma}
The truncated fidelity function $F(x,y)$ satisfies $F(x,y) \leq 1$ and is concave in $x$ and $y$.
\end{lemma}

Additionally, we prove the following.
\begin{lemma} \label{lem:boundfidelity}
    For all $x_0, y_0 \in [0,1]$, $\Delta \in [0, 1 - x_0]$ and for all $K > 0$, it holds that
    \[F(x_0 + \Delta, y_0) - K \Delta^2 \leq F(x_0, y_0) + O(K^{-1/3}). \]
\end{lemma}
\begin{proof}
We have:
\begin{align*}
    F(x_0 + \Delta, y_0) - F(x_0, y_0) &\leq (2y-1) \Delta + 2 \sqrt{y(1-y)}\left(\sqrt{(x_0+\Delta)(1-x_0-\Delta)} - \sqrt{x_0(1-x_0)}\right) \\
                                       &\leq \Delta + \sqrt{(x_0+\Delta)(1-x_0-\Delta)} - \sqrt{x_0(1-x_0)}\\
                                       &\leq \Delta + \sqrt{\Delta(1-\Delta)} \\
                                       &\leq 2\sqrt{\Delta}
\end{align*}
where the first inequality is obtained by substituting the expressions for $F$ in the case of $x\leq y$ (which is certainly an upper bound because $F(x,y) \leq 1$ everywhere), and the next two are obtained by optimizing over $y_0$ and $x_0$, respectively, while holding $\Delta$ fixed. This shows that 
\[F(x_0 + \Delta, y_0) - K\Delta^2 \leq F(x_0, y_0) + 2\sqrt{\Delta} - K\Delta^2.\]
For a fixed value of $K$, the maximum value of the right-hand side over $ \Delta \in [0,1]$ is $O(K^{-1/3})$, yielding the result.
\end{proof}

\begin{definition}[Min-entropy and Rényi $2$-entropy]

For a quantum state $\rho^A$ on register $A$, its min-entropy and Rényi $2$-entropy are given by
\[ H_{\min}(A)_\rho = -\log \|\rho^A\|_\infty, \quad \quad H_2(A)_\rho = -\log\|\rho^A\|_2.\]
If $\rho= \sum_ip_i \state{\psi_i}$, $H_{\min}(A)_\rho = -\log(\max_i p_i)$, and $H_2(A)_\rho = -\log\left(\sum_ip_i^2\right)$, which are the classical definitions of min-entropy and Rényi $2$-entropy of a probability distribution. For a classical probability distribution $p=\{p_i\}_i$, we use $H_{\min}(p)$ to denote its min-entropy.
\end{definition}

\begin{definition}[Conditional min-entropy and Rényi $2$-entropy]
For a state $\rho^{AB}$ on registers $A$ and $B$, the min-entropy of $A$ conditioned on $B$ with respect to $\rho$ is given by
\[ H_{\min}(A|B)_\rho = -\log\left(\inf_{\sigma^B, \lambda}\{\lambda: \rho^{AB} \preceq \lambda I^A\otimes \sigma^B\}\right).\]
The conditional Rényi $2$-entropy is given by
\[ H_2(A|B)_\rho = -\log\left(\inf_{\sigma^B}\Tr\left(\left((\sigma_B)^{-1/4}\rho^{AB}(\sigma^B)^{-1/4}\right)^2\right)\right).\]
It can be seen that these reduce to the definitions of $H_{\min}(A)_\rho$ and $H_2(A)_\rho$ respectively when the register $B$ is empty.
\end{definition}

The following fact about the two conditional entropies is not difficult to see.
\begin{fact}\label{fact:entropy-ineq}
For all states $\rho^{AB}$, $H_2(A|B)_\rho \geq H_{\min}(A|B)_\rho$.
\end{fact}

\begin{definition}[Smoothed conditional min-entropy]
For $\rho^{AB}$ on $AB$, the $\eps$-smoothed min-entropy of $A$ conditioned on $B$ is defined as
\[ H^\eps_{\min}(A|B)_\rho = \sup_{\rho'^{AB}: \|\rho^{AB} - \rho'^{AB}\|_\tr \leq \eps}H_{\min}(A|B)_{\rho'}.\]
Without conditioning, the smoothed min-entropy of $A$ is simply
\[ H^\eps_{\min}(A)_\rho = \sup_{\rho': \|\rho'-\rho\|_\tr \leq \eps}H_{\min}(A)_{\rho'},\]
which for classical distributions is
\[ H^\eps_{\min}(p) = \sup_{q: \|p-q\|_1 \leq 2\eps}H_{\min}(q).\]
Similarly, the $\eps$-smoothed Rényi-$2$ entropy of $A$ conditioned on $B$ is
\[ H^\eps_{2}(A|B)_\rho = \sup_{\rho'^{AB}: \|\rho^{AB} - \rho'^{AB}\|_\tr \leq \eps}H_{2}(A|B)_{\rho'}.\]
\end{definition}

We have the following equivalent characterization of classical smoothed min-entropy.
\begin{lemma} \label{lem:minentropydef}
Given a probability distribution $p=\{p_i\}_i$, for all $\delta \in [0,1]$, we have that $H^{\delta}_{\min}(p) = \log(1/S^*)$, where
\[S^* = \inf\left\{S : \sum_{i: \; p_i > S} (p_i - S) \leq \delta\right\}\]
\end{lemma}

\begin{proof}
We first show that given the optimal value $S^*$, there is a distribution $q$ such that $\|q-p\|_1 \leq 2\delta$, and $H_{\min}(q) = \log(1/S^*)$. In fact $q$ will just be the distribution that has $q_i=S^*$ for all $i$. Since $\sum_iq_i=\sum_ip_i=1$, we have,
\[ \sum_i|q_i-p_i| = 2\sum_{i: p_i > q_i}(p_i-q_i) = 2\sum_{i: p_i > S^*}(p_i-S^*) = 2\delta.\]

To show the other direction, suppose we have $q$ such that $\|q-p\|_1 \leq 2\delta$ and $H_{\min}(q)=\log(1/S')$. Since $q_i \leq S'$ for all $i$, we have,
\[ \sum_{i: p_i > S'}(p_i-S') \leq \sum_{i: p_i > S'}(p_i-q_i) \leq \sum_{i:p_i>q_i}(p_i-q_i) \leq \delta.\]
This completes the proof.
\end{proof}

\subsection{Concentration inequalities}
We will need to use a number of classical and quantum concentration inequalities, which we list here.

\begin{fact}[Weak law of large numbers] \label{fact:weaklaw}
Let $X$ be a real-valued random variable with $\mathbb{E}[|X|] < \infty$. Let $X_1, \ldots, X_n$ be i.i.d. realizations of $X$. Then, for all $\eps > 0$, it holds that
\[\lim_{n \to \infty} \Pr\left[\left|\frac{1}{n} \sum_{i=1}^nX_i - \mathbb{E}[X]\right| > \varepsilon\right] = 0.\]
\end{fact}

We have the following standard facts about the Gaussian distribution:
\begin{fact} \label{lem:gaussianbound}
    For all $c > 0$ and $a > 0$,
    \[\int_a^\infty \exp(-ct^2)dt \leq \frac{1}{2ca} \exp(-ca^2).\]
\end{fact}
\begin{fact} \label{fact:gaussianintegral}
    For all $c > 0$,
    \[\int_0^\infty \exp(-ct^2)dt = \sqrt{\frac{\pi}{4c}}.\]
\end{fact}

We also have the following property of sub-Gaussian random variables:
\begin{fact}[\cite{vershynin2018hdp}, Proposition 2.6.1] \label{prop:subgaussian}
Let $X$ be a random variable. The following properties are equivalent, with the parameters $K_i>0$ differing by at most an absolute constant factor.

\begin{enumerate}
\item[(i)] There exists $K_1>0$ such that
\[
\Pr\bigl[|X|\ge t\bigr] \le 2\exp\!\left(-\frac{t^2}{K_1^2}\right)
\quad \text{for all } t\ge 0.
\]

\item[(ii)]There exists $K_2>0$ such that
\[
\bigl(\mathbb{E}[|X|^r]\bigr)^{1/r} \le K_2 \sqrt{r}
\quad \text{for all } r\ge 1.
\]
\end{enumerate}
\end{fact}

\begin{lemma} \label{lem:expectationbound}
Let $X$ be a probability space, and let $p$ and $q$ be probability distributions on the space with $q(x) \leq K p(x)$ for all $x$, for some $K \geq 2$. Letting $r = \log K$, for all bounded measurable functions $f: X \mapsto \mathbb{R}_{\ge 0}$, we have
\[\mathbb{E}_{x\sim q}\left[f(x)\right] \leq 2  \left(\mathbb{E}_{x\sim p}\left[(f(x))^r\right]\right)^{1/r}\]
\end{lemma}
\begin{proof}
Let $w(x) = \frac{q(x)}{p(x)}$, so that $w(x) \leq K$ for all $x$ by assumption. Given $r \geq 1$,  take $r' \geq 1$ to be its H\"older conjugate, i.e. such that $1/r + 1/r' = 1$. We have:
\begin{align*}
\mathbb{E}_q[f(x)] &= \mathbb{E}_p[w(x)f(x)]\\
                &\leq \left(\mathbb{E}_p[(w(x))^{r'}]\right)^{\frac{1}{r'}} \left(\mathbb{E}_p[(f(x))^r]\right)^{\frac{1}{r}} \\
                &\leq \left(\mathbb{E}_p[K^{r'-1} w(x)]\right)^{\frac{1}{r'}} \left(\mathbb{E}_p[(f(x))^r]\right)^{\frac{1}{r}}\\
                &= K^{\frac{1}{r}} \left(\mathbb{E}_p[(f(x))^r]\right)^{\frac{1}{r}}\\
\end{align*}
where the expectation form of H\"older's inequality was applied to obtain the first inequality. Setting $r = \log K$, we see that the statement of the lemma is obtained.
\end{proof}

\begin{fact}[Operator Chernoff bound, \cite{ahlswede2001}] \label{fact:matrixchernoff}
Let $X_1,\dots,X_m$ be i.i.d. $d \times d$ PSD matrices such that $\|X_i\| \leq 1$ almost surely. 
Let
\[
A := \mathbb E[X_j],
\qquad \alpha = \lambda_{\min}(A) 
\]
Then, for all $0 < \eps < 1/2$, 
\[
\Pr\!\left[
(1-\eps)A \preceq \frac{1}{m}\sum_{j=1}^m X_j \preceq (1+\eps)A
\right]
\;\ge\;
1 - 2d \exp\!\left(-\frac{m\alpha\eps^2}{2\ln 2}\right).
\]
\end{fact}

We will need some results on the concentration of measure on the unitary group, which holds for Lipschitz functions. We define what Lipschitz functions are first, and state a useful property of such functions.
\begin{definition}
Given two metric spaces $(X, d_X)$ and $(Y, d_Y)$ (where $d_X$ is the metric on $X$ and $d_Y$ on $Y$), a function $f:X\to Y$ is called $\kappa$-Lipschitz with respect to these norms if
\[ d_Y(f(x_1), f(x_2)) \leq \kappa\cdot d_X(x_1, x_2).\]
\end{definition}
For our applications, we will be dealing with functions from operators to $\mathbb{R}$; we will use the Schatten $2$-norm on operators (i.e. the Frobenius norm) and the absolute value on the reals.
\begin{fact} \label{fac:extlipschitz} Let $X$ be a metric space and let $\{f_i\}_{i \in I}$ be a collection of functions from $X$ into $\mathbb{R}$ which are all $\kappa$-Lipschitz. We have that the functions
\[f(x) = \sup_i(f_i(x))\]
\[g(x) = \inf_i(f_i(x))\]
are also $\kappa$-Lipschitz.
\end{fact}

We can now give the following concentration inequality on the unitary group, which uses Lipschitz constants.
\begin{fact}[\cite{Meckes_2019}, Theorem 5.17] \label{thm:concentration_}
Let $U(d)$ denote the group of $d\times d$ unitary matrices.  
Let $f : U(d)\to\mathbb{R}$ be a $\kappa$-Lipschitz function with respect to the Schatten $2$-norm on $U(d)$.
There exists a universal constant $c>0$ such that, for all $t>0$,
\[
\Pr\bigl[|f(U)-\mathbb{E}[f(U)]| \ge t\bigr]
\;\le\;
\exp\,\bigl(-c\,d\,t^2/\kappa^2\bigr).
\]
\end{fact}
From this result we can derive the corresponding concentration of measure result on $G(d,k)$.

\begin{lemma} \label{prop:concentration}
    Let the projector $P$ be drawn uniformly from $G(d,k)$. Let $f: G(d,k) \to \mathbb{R}$ be a $\kappa$-Lipschitz function with respect to the Schatten $2$-norm. There exists a universal constant $c' > 0$ such that, for all $t>0$,
\[
\Pr\bigl[|f(P)-\mathbb{E}[f(P)]| \ge t\bigr]
\;\le\;
\exp\,\bigl(-c\,d\,t^2/\kappa^2\bigr).
\]
\end{lemma}
\begin{proof}
A uniformly random $P$ is obtained by sampling a uniformly random $U \in U(d)$ and setting $P = U_k U_k^{\dagger}$, where $U_k$ refers to the first $k$ columns of $U$. This parameterization will let us apply Fact \ref{thm:concentration_}. Given $U, V \in U(d)$ with $U \neq V$, we have
\begin{align*}
|f(U_k U_k^\dagger) - f(V_k V_k^\dagger)| &\leq \kappa \|U_k U_k^\dagger - V_k V_k^\dagger\|_2\\
                                          &= \kappa\|U_k U_k^\dagger - U_k V_k^\dagger + U_k V_k^\dagger - V_k V_k^\dagger\|_2 \\
                                          &\leq \kappa\|U_k\| \|U_k^\dagger - V_k^\dagger\|_2 + \kappa \|V_k^\dagger\| \|U_k - V_k\|_2 \\
                                          &= 2\kappa \|U_k - V_k\|_2
\end{align*}
This shows that the map $g: U(d) \to \mathbb{R}$ defined by $g(U) = f(U_k U_k^\dagger)$ is $2\kappa$-Lipschitz. The result then follows from Fact \ref{thm:concentration_}.
\end{proof}

\subsection{Results about epsilon nets}
With respect to the distance metric on $G(d,k)$ we have defined previously, the following bounds are known on the covering number $N(G(d,k), \varepsilon)$, i.e. the size of the smallest $\varepsilon$-net.

\begin{fact}[\cite{Szalek}, see also \cite{Pajor1998}, Proposition 8]\label{lem:net1}
There exist universal constants $C > c > 0$ such that, for all $k \leq d/2$ and all $\varepsilon > 0$, setting $m=2k(d-k)$, 
\[\left(\frac{c}{\varepsilon}\right)^m \leq N(G(d,k), \varepsilon) \leq \left(\frac{C}{\varepsilon}\right)^m\]
\end{fact}

We also have the following result about $\eps$-nets on the $d$-dimensional complex sphere (set of $d$-dimensional pure quantum states), where the metric is the $2$-norm for vectors.
\begin{fact}[\cite{vershynin2018hdp}, Corollary 4.2.11]\label{fc:sphere-net}
There exists an $\eps$-net on the set of $d$-dimensional pure quantum states of size at most $\left(1+\frac{2}{\eps}\right)^{2d}$.
\end{fact}

Finally, we will be needing this result for estimating the spectral norm of a matrix using an $\varepsilon$-net on the sphere.
\begin{fact}[\cite{vershynin2018hdp}, Lemma 4.4.1] \label{lem:net2}
Let $A$ be an $m\times d$ matrix and let $\varepsilon\in[0,1)$. Then, for any
$\varepsilon$-net $\mathcal N$ of the sphere in $\mathbb{C}^d$, we have
\[
\sup_{x\in\mathcal N} \|Ax\|_2
\;\le\;
\|A\|_\infty
\;\le\;
\frac{1}{1-\varepsilon}\,
\sup_{x\in\mathcal N} \|Ax\|_2,
\]
where $\|Ax\|_2$ is the $2$-norm of the vector $Ax$.
\end{fact}


\subsection{Entanglement and LOCC}
\begin{fact}[Schmidt decomposition]
A pure entangled state $\ket{\psi}^{AB}$ on registers $A$ and $B$ can be written as
\[ \ket{\psi}^{AB} = \sum_i\lambda_i\ket{a_i}^A\ket{b_i}^B,\]
where $\ket{a_i}^A$ and $\ket{b_i}^A$ are orthonormal states and $\lambda_i$ are non-negative real numbers satisfying $\sum_i\lambda_i^2=1$, which are unique up to reordering. The above expression is called the Schmidt decomposition of the state $\ket{\psi}$. $\{\lambda_i^2\}_i$ is a probability distribution called its Schmidt spectrum, and this is equal to the eigenvalues of the marginal states $\psi^A$ and $\psi^B$. In what follows, the Schmidt spectrum will always be assumed to be sorted in descending order.
\end{fact}

\begin{definition}[Entanglement min-entropy]
For a pure entangled state $\ket{\psi}^{AB}$, its $\eps$-smoothed entanglement min-entropy is equal to either $H^\eps_{\min}(A)_\psi$ or $H^\eps_{\min}(B)_\psi$. It can be seen that both these quantities are equal to the $\eps$-smoothed min-entropy of its Schmidt spectrum.
\end{definition}
Entanglement min-entropy is a lower bound on distillable entanglement from a pure state, as we shall see below.

\begin{definition}[Entanglement distillation]
An entanglement distillation protocol between two parties Alice and Bob, who start with an arbitrary initial shared entangled state $\rho^{AB}$, is a protocol utilizing only local operations and classical communication (LOCC) between the parties, which ends with Alice and Bob sharing the state $\sigma^{AB}$ satisfying
\[ \mathsf{F}(\sigma^{AB}, \state{\Phi_+}^{AB}) \geq 1-\eps,\]
where $\state{\Phi_+}^{AB}$ is the maximally entangled state on $AB$ of rank $d$, for some $d$. The parameter $\eps$ is called the error of the protocol.
\end{definition}

\begin{fact}[\cite{WTB17}, Proposition 21]\label{fc:distillation-lb}
For $\eta \in [0, \sqrt{\eps})$, there exists a one-way $\eps$-error entanglement distillation protocol between Alice and Bob who share the entangled state $\rho^{AB}$, utilizing classical communication only from Alice to Bob, yielding a maximally entangled state of rank $d$, for
\[ \log d = H_{\min}^{\sqrt{\eps} - \eta}(A|E)_\rho - 4\log(1/\eta).\]
Here the register $E$ is the purifying register in an arbitrary purification $\ket{\rho}^{ABE}$ of $\rho^{AB}$.
\end{fact}
If $\rho^{AB}$ is a pure state, the entanglement distillation procedure from this theorem yields $H^{\sqrt{\eps}-\eta}_{\min}(A)_\rho-4\log(1/\eta)$ many EPR pairs. We'll also note that inspecting the proof of Proposition 21 in \cite{WTB17} reveals that their protocol works fine with $H_2$ instead of $H_{\min}$. So the actual lower bound on one-way distillable entanglement is $H_2^{\sqrt{\eps}-\eta}(A|E)_\rho - 4\log(1/\eta)$ for mixed states, and  $H_2^{\sqrt{\eps}-\eta}(A)_\rho - 4\log(1/\eta)$ for pure states.

Finally, we will use a result characterizing when a pure entangled state can be transformed into an ensemble of pure entangled states via LOCC. In this scenario, we are considering implementing a measurement that gives the $i$-th outcome with probability $p_i$, and the post-measurement joint state when the $i$-th outcome is obtained is $\ket{\phi_i}$. We use $\{(p_i, \state{\phi_i})\}_i$ to denote the outcome ensemble.

To state the result, we need the concept of majorization.
\begin{definition}
Let $x$ and $y$ be two $d$-dimensional vectors whose components are sorted in decreasing order, i.e., $x_1 \geq x_2 \geq \ldots \geq x_d$. We say that $x$ majorizes $y$, denoted by $x \succeq y$,\footnote{Note that we have used the same notation for operator inequalities, but it will be clear from context which one we mean.} iff
\[ \sum_{i=1}^k x_i \geq \sum_{i=1}^ky_i \quad \forall k \in [d].\]
For vectors that are not arranged in descending order, we need to arrange them in descending order first before checking if the majorization condition is satisfied.
\end{definition}

\begin{fact}[\cite{Jonathan_1999}, see also \cite{NielsenVidal2001}, Theorem 16]\label{thm:jonathanplenio}
Take $\ket{\psi}$ to be a pure state and take $\{(p_i, \state{\phi_i})\}$ to be an ensemble. Letting $\lambda$ be the Schmidt spectrum of $\psi$ and the $\mu_i$ be those of the $\state{\phi_i}$ (all sorted in descending order), the transformation $\state{\psi} \rightarrow \{(p_i, \state{\phi_i})\}$ can be accomplished with LOCC if and only if:
\[\lambda \preceq \sum_i p_i \mu_i.\]
\end{fact}

\subsection{Remote state preparation}
In a remote state preparation setting, Alice gets a classical description of a quantum state. Alice and Bob share entanglement and do local operations and classical communication on their halves of it. The requirement is that at the end of the protocol, Bob should have a state close to Alice's input state. This closeness can be measured in different ways, either in the worst-case or on average.

In general, remote state preparation protocols could have two-way communication, and this may be especially important in cases when the shared entangled state between Alice and Bob is mixed. This is because there exist mixed entangled states from which one can distill EPR pairs with two-way communication, but not one-way communication \cite{BDSW96}. Correspondingly, there may be mixed entangled states with which RSP with efficient communication is possible in the two-way case but not the one-way case.

However, it is more common to consider RSP protocols with one-way communication in the literature, and we shall stick to this too, for convenience. In any case, we only consider pure entangled states in this work, and allowing two-way communication may not make much difference for pure states. However, our lower bound result does also work for protocols with two-way communication, as long as we replace the size of Alice's message to Bob with the total size of the transcript. This is easy to see for our communication lower bound, but more tedious for the entanglement lower bound.

We describe the general form of an RSP protocol for preparing flat states $\frac{P}{k}$ for $P \in G(d,k)$ below. We call this a $(d,k)$-RSP protocol. We also introduce three different notions of error for such protocols.

\begin{center}
\begin{mdframed}[nobreak=true]
\textbf{Parameters.}
\begin{itemize}
  \item Target dimension: \(d\)
  \item Target rank: \(k\)
  \item Classical message length: $m$
\end{itemize}

\bigskip
\textbf{Input.}
Alice is given a classical description of a rank-$k$ projector $P$ on $\mathbb{C}^d$.

\medskip
\textbf{Protocol.}
Alice and Bob initially share a pure entangled state \(\ket{\sigma}^{AB}\)
\begin{enumerate}
  \item Alice performs an isometry $U_P: A \mapsto A'C$, with the register $C$ being of dimension $2^m$.
  \item Alice measures $C$ and sends the measurement outcome $c \in \{0,1\}^m$ to Bob.
  \item Bob applies an isometry \(V_c : B \rightarrow B_1 B_2\), with $B_2$ being of dimension $d$. For message $c$, the final state in the registers $A'B_1B_2$ is $\ket{\chi_{P,c}}^{A'B_1B_2}$, and its marginal on the target register $B_2$ is $\chi^{B_2}_{P,c}$.
\end{enumerate}

\textbf{Error measures.}
\begin{enumerate}
\item The worst-case error $\varepsilon_w$ of the protocol is the smallest $\varepsilon$ such that, for all $P$, 
\[\mathbb{E}_c(\Vert\chi^{B_2}_{P,c} - P/k\Vert_\tr) \leq \varepsilon.\]
\item The average-case error $\varepsilon_a$ of the protocol is
\[\mathbb{E}_{P, c} (\Vert\chi^{B_2}_{P,c} - P/k\Vert_\tr).\]
\item The relaxed average-case error $\varepsilon_r$ of the protocol is 
\[\mathbb{E}_{P, c} \Tr((I-P)\chi^{B_2}_{P,c}).\]
\end{enumerate}
\end{mdframed}
\captionsetup{hypcap=false}
\captionof{figure}{Formal description of a $(d,k)$-RSP protocol with pure shared entanglement}
\label{fig:protocol}
\end{center}
We note that in our protocol description, the restriction to isometries is done without loss of generality as general quantum operations can be purified into isometries through the addition of an ancilla, which our model allows for. Also, it is apparent that $\varepsilon_w \geq \varepsilon_a$, and it follows from point 3 of Fact \ref{fac:distanceproperties} that $\varepsilon_a \geq \varepsilon_r$.

\section{Efficient protocols for remote state preparing flat states}

In this section, we will give two protocols for remote state preparing flat states, achieving any desired worst-case error $\eps > 0$. The protocols achieve communication and entanglement costs that are roughly similar up to $\eps$-dependent parameters, and it will be shown that they are essentially resource-optimal.

Ideally we want RSP protocols to succeed in the worst case. But when designing RSP protocols, one finds that it is more convenient to work with the average-case error. Indeed, our two base constructions will only have average-case correctness. In order to obtain protocols with worst-case error from protocols with average-case error, we first give an average-case to worst-case reduction in Section \ref{sec:av-to-worst}. Then we give two different protocols with average-case error in Sections \ref{sec:protocol-1} and \ref{sec:protocol-2}. Together, the results of these three theorems imply Theorem \ref{thm:rsp-upper}, whose statement is recalled below.
\upperbound*

\subsection{Converting a good average-case protocol into a good worst-case one}\label{sec:av-to-worst}
In this section, we show that for $(d,k)$-RSP protocols, the notions of average-case correctness over the uniform distribution in $G(d,k)$, and worst-case correctness, are essentially equivalent. This means that a protocol with a certain average-case error $\varepsilon_a$ can be turned into a new protocol with worst-case error arbitrarily close to $\varepsilon_a$ at the cost of a modest increase in the communication. 

The main result of this section is Theorem \ref{thm:worstcase}, whose statement is recalled below.
\avtoworst*

Given a protocol $\mathcal{P}$ which achieves average-case error $\varepsilon_a$, we will consider a modified protocol $\mathcal{P}'$ as given in Figure \ref{fig:prot1}. Given $\delta > 0$, we want to pick $N$ to be such that, for some such choice of unitaries $\mathcal{U}$, the worst-case error of the resulting protocol is at most $\varepsilon_a + \delta$. Theorem \ref{thm:N-avtoworst} below shows how to do this, which proves Theorem \ref{thm:worstcase}.
\begin{center}
\begin{mdframed}[nobreak=true]
\textbf{Given.}
\begin{itemize}
\item A $(d,k)$-protocol $\mathcal{P}$
\item An error tolerance $\delta > 0$
\item A collection $\mathcal{U} = \{U_1,\ldots,U_N\}$ of unitaries on $\mathbb{C}^d$
\item A $\delta/4$-net $\mathcal{N}$ over $G(d,k)$ of minimal size
\end{itemize}

\textbf{Protocol.}
Alice and Bob share the same entangled state as in $\mathcal{P}$.
\begin{enumerate}
  \item Given $P \in G(d,k)$, for every $i \in [N]$, Alice sets $\tilde{P_i} = U_i P U_i^\dagger$.
  \item For each $i\in [N]$, Alice finds $P_i \in \mathcal{N} \cap B(\tilde{P_i}, \frac{\delta}{4})$ which minimizes the error $\eps(P_i)$ of $\mathcal{P}$ on input $P_i$. Let $\varepsilon_{\delta/4}(i)$ be this minimal error for $i$. Then Alice finds the minimal $i\in [N]$ such that $\varepsilon_{\delta/4}(i)$ is minimized; let $i^*$ be this minimal $i$.
  \item Alice and Bob run $\mathcal{P}$ on input $P_{i^*}$.
  \item Alice sends $i^*$ to Bob, who performs $U_{i^*}^\dagger$ on $B_2$.
\end{enumerate}
\end{mdframed}
\captionsetup{hypcap=false}
\captionof{figure}{Protocol $\mathcal{P}'$ with worst-case error, given $\mathcal{P}$ with average-case error}
\label{fig:prot1}
\end{center}

In order to prove Theorem \ref{thm:N-avtoworst}, we start by proving a few lemmas.
\begin{lemma} \label{lem:tracelipschitz}
    Let $P_0$ be a fixed projector in $G(d,k)$, and define the function $f: G(d,k) \to \mathbb{R}$ by:
    \[f(P) = \|P/k - P_0/k\|_\tr.\]
     Then, $f$ is $\left(\frac{1}{\sqrt{2k}}\right)$-Lipschitz.
\end{lemma}
\begin{proof}
For all $P, Q \in G(d,k)$, it holds
\begin{align*}
    |f(P) - f(Q)| &\leq \frac{1}{2k} \|P-Q\|_1\\
                  &\leq \frac{1}{2k} \sqrt{2k} \|P-Q\|_2
\end{align*}
where the reverse triangle inequality was used for the first inequality and the fact that $P-Q$ has rank at most $2k$ was used in the second.
\end{proof}
This lets us derive the following isoperimetric inequality for $G(d,k)$. 
\begin{lemma} \label{lem:isoperimetry}
Let $A$ be a measurable subset of $G(d,k)$ and let $\mu(A)$ be its measure under the Haar measure. Define the distance function $f_A: G(d,k) \to \mathbb{R}$ by:
\[f_A(P) = \inf_{Q \in A}(\|P/k - Q/k\|_\tr).\]
There exists a universal constant $c'' > 0$ such that, for all $t > 0$,
\[ \Pr[f_A(P) \geq t] \leq \exp(-c''dkt^2 \mu(A)^2).\]
\end{lemma}
\begin{proof}
Note first that it follows from Lemma \ref{lem:tracelipschitz} along with Fact \ref{fac:extlipschitz} that $f_A$ is $\left(\frac{1}{\sqrt{2k}}\right)$-Lipschitz. Given $t > 0$, define the function $g_t: G(d,k) \to \mathbb{R}$ by:
\[g_t(P) = \max\left(0, 1 - \frac{f_A(P)}{t}\right).\]
Note that $g_t(P) = 1$ if and only if $P \in \overline{A}$ and $g_t(P) = 0$ if and only if $f_A(P) \geq t$. We have that $g_t$ is $\left(\frac{1}{\sqrt{2k} t}\right)$-Lipschitz in view of Fact \ref{fac:extlipschitz}. We also have that $\mathbb{E}[g_t(P)] \geq \mu(A)$, since in the set $A$, $f_A(P)=0$, and $g_t(P)$ takes value $1$. This all gives:
\begin{align*}
    \Pr[f_A(P) \geq t] &= \Pr[g_t(P) = 0] \\ 
                              &\leq \Pr[|g_t(P) - \mathbb{E}[g_t(P)]| \geq \mu(A)]\\
                              &\leq \exp(-2c'dkt^2 \mu(A)^2)
\end{align*}
where Lemma \ref{prop:concentration} was applied in the last line. The statement follows.
\end{proof}

With this result in hand, we prove Theorem \ref{thm:N-avtoworst} below, which lets us pick an appropriate choice of $N$ for Theorem \ref{thm:worstcase}.
\begin{theorem}\label{thm:N-avtoworst}
    In the protocol $\mathcal{P}'$ in Figure \ref{fig:prot1}, it is enough to take $N = O\left(\frac{\log 1/\delta}{\delta^4}\right)$.
\end{theorem}
\begin{proof}
    Let $\varepsilon(P)$ be the error of $\mathcal{P}$ on input $P$, and set:
    \[\varepsilon_{\delta/4}(P) = \inf_{Q \in \mathcal{N} \cap B(P, \delta/4)}(\varepsilon(Q)).\] 
    
   We claim that $\mathcal{P}'$ will have worst-case error at most $\varepsilon_a + \delta$ if the unitaries $\mathcal{U}$ are such that, for every $P' \in \mathcal{N}$,
    \begin{equation} \label{eq:worstcasecondition}
    \min_{i \in [N]} \varepsilon_{\frac{\delta}{4}}(U_i P' U_i^\dagger) < \varepsilon_a + \frac{\delta}{2}.
    \end{equation}
    On input $P$, $\mathcal{P}'$ then prepares $P'\in \mathcal{N}$ that is $\delta/2$-close to $P$ up to distance $\eps_a+\delta/2$. Therefore, the overall distance between $P$ and Bob's output on $P$ is $\eps_a+\delta$.
    
    The existence of such a choice of unitaries $\mathcal{U}$ will now be shown using the probabilistic method. It is enough to show that for a specific choice of $P \in \mathcal{N}$, if $\mathcal{U}$ is chosen at random by selecting each $U_i$ independently from the Haar measure, the probability that equation \eqref{eq:worstcasecondition} is not satisfied is strictly less than $\frac{1}{|\mathcal{N}|}$. This is because we can take a union bound over $\mathcal{N}$ to say that the probability (over the random choice of $\mathcal{U}$) that the condition does not hold for every choice of $P\in \mathcal{N}$ is smaller than $1$, which means there exists a choice of $\mathcal{U}$ for which the condition holds for every $P\in\mathcal{N}$.
    
    Assuming $U_1, \ldots, U_N$ are thus chosen, fixing a $P' \in \mathcal{N}$, by the definition of Haar measure, setting $P_i = U_i P' U_i^\dagger$, we see that the $P_i$ are independent and uniformly distributed in $G(d,k)$. Define the set $A \subseteq G(d,k)$ by:
    \[A = \left\{Q : \varepsilon_{\delta/4}(Q) < \varepsilon_a + \frac{\delta}{2}\right\}\]
    Markov's inequality can be seen to imply that $\mu(A) \geq \frac{\delta}{2}$. It suffices to bound the probability that, for all $i$, $f_A(P_i) > \frac{\delta}{2}$, where $f_A(P_i)$ is as defined in Lemma \ref{lem:isoperimetry}. Applying Lemma \ref{lem:isoperimetry}, for some universal constant $K$, this is at most
   \[
        \exp(-Kdk\delta^4 N).
    \]
    We have thus upper bounded the probability of $U_1, \ldots, U_N$ being bad for a single $P'\in \mathcal{N}$. We know from Fact \ref{lem:net1} that
    \[|\mathcal{N}| \leq \left(\frac{2C}{\delta}\right)^{2k(d-k)}.\]
    Setting $N=O\left(\frac{\log(1/\delta)}{\delta^4}\right)$ as in the statement of the theorem therefore makes the probability of $U_1, \ldots, U_N$ being bad for any $P'\in \mathcal{N}$ as small as needed.
\end{proof}

\subsection{First protocol: the Kraus operator protocol}\label{sec:protocol-1}
In this section, we will give our first protocol for RSP of flat states with average-case error, proving the following theorem.
\begin{theorem}\label{thm:rsp-prot-1}
For every $\eps_a > 0$, there exists a $(d,k)$-RSP protocol with average error $\eps_a$ in which Alice and Bob share the standard maximally entangled state with local dimension $d$, and Alice communicates $\log \frac{d}{k} + \log \log d + 2\log \frac{1}{\varepsilon_a} + O(1)$ bits to Bob. Moreover, the error in this protocol is only one-sided, i.e. either it fails and the failure is known to Alice and Bob, or it succeeds and the residual state on Bob's end is exactly $P/k$.
\end{theorem}
Theorem \ref{thm:rsp-prot-1} implies item \ref{item:prot-1} of Theorem \ref{thm:rsp-upper}. It is also possible to apply Theorem \ref{thm:worstcase} to the protocol in Theorem \ref{thm:rsp-prot-1} to get a worst-case correct protocol. Note however, that the fact that the protocol only makes one-sided error is not preserved after applying the average-case to worst-case reduction, so the resulting protocol is simply worse than the protocol in item \ref{item:prot-2} of Theorem \ref{thm:rsp-upper}.

This first protocol that we give is a straightforward generalization of the remote state preparation protocol for pure states that is given in 
\cite{bennett2005remote}, and basically reduces to exactly it when we set $k=1$. The only slight difference is that \cite{bennett2005remote} had to work a bit harder to ensure worst-case correctness, whereas we only need to control the average-case error, which makes the analysis somewhat simpler. As in the previous average-case to worst-case reduction, the protocol is specified by a collection of $N$ unitaries, for some appropriately chosen value of $N$, and is described in Figure \ref{fig:prot2}.

\begin{center}
\begin{mdframed}[nobreak=true]
\textbf{Parameters.}
A collection $\mathcal{U} = \{U_1,\ldots,U_N\}$ of unitaries on $\mathbb{C}^d$

\textbf{Protocol.}

Alice and Bob share the standard maximally entangled state $\ket{\Phi_+}$ on $\mathbb{C}^d \otimes \mathbb{C}^d$.
\begin{enumerate}
  \item Given $P \in G(d,k)$, Alice measures her share of $\ket{\Phi_+}$ according to the generalized measurement $\{M_e\} \cup \{M_i\}_{i \in [N]}$, where
\[M_e = \sqrt{I - M/\|M\|_\infty}\]
\[M_i= \sqrt{\frac{d}{kN \|M\|_\infty}} U_i \overline{P} U_i^\dagger\]
where
  \begin{equation}\label{eq:Kraus-def}
M = \frac{d}{kN} \sum_{i \in [N]} U_i \overline{P} U_i^\dagger
\end{equation}
  \item Alice sends her measurement outcome to Bob. If Bob receives $e$, he does nothing: if he receives $i \in [N]$, he applies $\overline{U_i}^\dagger$ to his share of the state, which can be checked to now have reduced density matrix exactly $P/k$. 
\end{enumerate}
\end{mdframed}
\captionsetup{hypcap=false}
\captionof{figure}{First average-case-correct $(d,k)$-RSP protocol}
\label{fig:prot2}
\end{center}

For the protocol in Figure \ref{fig:prot2} to have average-case error at most $\varepsilon_a$, we see that we must exhibit $N$ so that there exists a choice of unitaries $\mathcal{U}$ such that, if $P \in G(d,k)$ is sampled uniformly at random, and $M$ is as in the description of the protocol, then
\[\mathbb{E}_P\left[1 - \frac{1}{\|M\|_\infty}\right] < \varepsilon_a.\]
This will ensure that the average probability of the protocol producing the error outcome $e$ is at most $ \varepsilon_a$. This value of $N$ is exhibited in the following result.

\begin{theorem}
    Such a choice of $\mathcal{U}$ exists if we take $N = \Theta(\frac{d \log d }{k\varepsilon_a^2})$.
\end{theorem}

\begin{proof}
With the choice of $N$ in the theorem statement, the existence of an appropriate collection $\mathcal{U}$ is shown using the probabilistic method. For a given choice of $\mathcal{U}$ and a given choice of $P$, define $\Lambda(\mathcal{U}, P)$ to be the spectral norm of $M$ defined in equation \eqref{eq:Kraus-def} for a fixed $\mathcal{U}$ and $P$: 
\[ \Lambda(\mathcal{U}, P) = \left\| \frac{d}{kN} \sum_i U_i \overline{P} U_i^\dagger \right\|_\infty.\]

We want to show the existence of a collection of unitaries $\mathcal{U}$ such that
\[\mathbb{E}_P\left[1 - \frac{1}{\Lambda(\mathcal{U}, P)}\right] < \eps_a.\]
Note that since $\Lambda(\mathcal{U}, P) \geq 1$ always (since $M$ has trace $d$), we get
\[\mathbb{E}_P\left[1 - \frac{1}{\Lambda(\mathcal{U}, P)}\right] \leq \mathbb{E}_P[\Lambda(\mathcal{U}, P)] - 1.\]
We will pick the elements of $\mathcal{U}$ i.i.d. from the Haar measure, and show that 
\[\mathbb{E}_\mathcal{U}\mathbb{E}_P[\Lambda(\mathcal{U}, P)] < 1 + \eps_a.\]
This, by the probabilistic method, implies that there exists a choice of $\mathcal{U}$ such that $\mathbb{E}_P[\Lambda(\mathcal{U}, P)]$ is at most $1+\eps_a$, which gives us the required average probability of the error outcome $e$. Note that by the right-invariance of the Haar measure, taking $P_1, \ldots, P_N$ to be independent uniformly random elements of $G(d,k)$ and defining the random variable $X$ by 
\[X = \left\|\frac{d}{kN} \sum_i P_i\right\|_\infty\]
we have that
\[\mathbb{E}_{\mathcal{U}}\mathbb{E}_P[\Lambda(\mathcal{U}, P)] = \mathbb{E}[X].\]

This expectation can now be estimated using the operator Chernoff bound (Fact \ref{fact:matrixchernoff}). Firstly, it is easy to see that for any $i$, $\mathbb{E}_{P_i}\left[\frac{d}{k}P_i\right]$ is the identity operator. The operator Chernoff bound then gives that, for $K = \Theta(\frac{Nk}{d}) = \Theta(\frac{\log d}{\eps _a^2})$ and for all $t > 0$, 
\[\Pr\left[X > 1 + t\right] = \Pr_{P_1, \ldots, P_N}\left[\frac{d}{kN} \sum_i P_i \not \preceq \left(1 + t\right)I\right] \leq 2d\cdot\exp\left(-Kt^2\right).\]
Set $t_0 = \sqrt{\frac{\ln d}{K}}$. Then,
\begin{align*}
    \mathbb{E}[X] &= 1 + \int_0^\infty \Pr[X > 1+t] dt \\
    &\leq 1 + t_0 + \int_{t_0}^\infty \Pr[X > 1 + t] dt \\
     &\leq 1 + t_0 + \int_{t_0}^\infty 2d\cdot\exp(-Kt^2) dt \\
     &\leq 1 + t_0 + \frac{d}{Kt_0} \exp(-Kt_0^2)\\
     &= 1 + t_0 + \frac{1}{K t_0}
\end{align*}
where Fact \ref{lem:gaussianbound} was used in the second to last inequality. It can be seen that our choice of $N$ makes this quantity at most $1 + \eps_a$, as desired.
\end{proof}

\subsection{Second protocol: Damped rejection sampling}\label{sec:protocol-2}

The protocol we gave in Section \ref{sec:protocol-1} was simple to analyze, had zero error and was quite economical when it came to the entanglement cost. However, its communication cost was not optimal as it included a $\log\log d$ term. In this subsection, we describe a more sophisticated $(d, k)$-RSP protocol whose communication cost matches that of the best-known protocol due to \cite{jain2005} up to $\eps$-dependent terms (which itself is a generalization of a protocol for pure states due to \cite{Bennett_2001}). The \cite{jain2005} rejection sampling uses an extravagant amount of entanglement, whereas the entanglement cost of our protocol, which is a refined version of rejection sampling, is nearly optimal.

This second protocol, given in Figure \ref{fig:prot3}, will prove the following theorem.

\begin{theorem}\label{thm:rsp-prot-2}
For every $\eps_a>0$, there exists a $(d,k)$-RSP protocol with average error $\eps_a$, in which Alice and Bob share the maximally entangled state with local dimension $d$, and Alice communicates $\log\frac{d}{k} + 2\log\frac{1}{\eps_a} + \log\ln\frac{1}{\eps_a} + O(1)$ bits to Bob.
\end{theorem}
Together, Theorems \ref{thm:rsp-prot-2} and \ref{thm:worstcase} imply item \ref{item:prot-2} of Theorem \ref{thm:rsp-upper}.

\begin{center}
\begin{mdframed}[nobreak=true]
\textbf{Parameters.}
\begin{itemize}
\item A natural number $N$ and a damping parameter $\eta \in (0,1]$
\item A collection $\mathcal{U} = \{U_1,...,U_N\}$ of unitaries on $\mathbb{C}^{d}$
\end{itemize}

\textbf{Protocol.}

Alice and Bob share the standard maximally entangled state $\ket{\Phi_+}^{AB}$, where $A$ and $B$ are both of dimension $d$.
\begin{enumerate}
  \item Given input $P$, Alice iterates for $i=1,\ldots,N$:
  \begin{enumerate}[label=\arabic{enumi}.\arabic*]
    \item 
    She measures register $A$ of the state according to the POVM with Kraus operators given by $\{\sqrt{\eta} Q_i, \sqrt{I - \eta Q_i}\}$, where $Q_i = U_i\overline{P}U_i^\dagger$.
    \item If the first outcome was obtained, she sends $i$ to Bob and exits the loop. If the second outcome was obtained and $i=N$, Alice sends $N+1$ to Bob. 
  \end{enumerate}
  \item If Bob receives $i \in [N]$ from Alice, he applies the unitary $U_i^T$ to $B$, and makes $B$ into the maximally mixed state otherwise. The register $B$ then contains the approximation of $P/k$ that is produced by the protocol. 
  \end{enumerate}
\end{mdframed}
\captionsetup{hypcap=false}
\captionof{figure}{Second average-case protocol for RSP of flat states}
\label{fig:prot3}
\end{center}

To analyze our new protocol, let us introduce notation. For $i \in \{0,1,...,N\}$, we will write $\ket{\psi_i}_{AB}$ to denote the shared state between Alice and Bob after $i$ failed iterations, and we will write $p_i$ to denote the probability that iteration $i$ succeeds given that iterations $1,...,i-1$ failed. For $i \in \{1,...,N+1\}$, we will write $\ket{\phi_i}$ to denote the shared state between Alice and Bob given that iteration $i$ succeeded, but before Bob applied his unitary correction. Also, we will write 
\[\rho_i =\Tr_A (I \otimes U_i^T) \ket{\phi_i} \bra{\phi_i} (I \otimes \overline{U_i})), \]
and we will write $\ket{\tilde{\phi}_i} = \sqrt{\frac{d}{k}}(Q_i \otimes I) \ket{\Phi_+}$. Note that, using Fact \ref{fact:maximallyentangled}, 
\[\Tr_A (I \otimes U_i^T) \ket{\tilde{\phi}_i} \bra{\tilde{\phi}_i} (I \otimes \overline{U_i})) = \frac{P}{k}.\]
For a given choice of unitaries $\mathcal{U}$ and Alice input $P$, we will write $\mathcal{E}(\mathcal{U}, P)$ to mean the expected error of the protocol measured in trace distance. For a given $\eps_a > 0$, we wish to exhibit values of $\eta$ and $N$ such that there exists a choice of unitaries $\mathcal{U}$ with the property that:
\[\mathbb{E}_P[\mathcal{E}(\mathcal{U}, P)] \leq \eps_a.\]
Actually, it will be more convenient for us to work with the expected fidelity $\mathcal{F}(\mathcal{U}, P)$ of the protocol instead. Note that since the function $\sqrt{1-x}$ is concave on $[0,1]$, Jensen's inequality together with the Fuchs-van de Graaf inequality (item \ref{item:fvdg} of Fact \ref{fac:distanceproperties}) imply that:
\[\mathbb{E}_P[\mathcal{E}(\mathcal{U}, P)] \leq \sqrt{1-\mathbb{E}_P[\mathcal{F}(\mathcal{U}, P)]}.\]
Hence, it suffices to show the existence of the unitaries $\mathcal{U}$ such that
\[\mathbb{E}_P(\mathcal{F}(\mathcal{U}, P)) \geq 1 - \eps_a^2.\]
The law of total expectation gives:
\begin{align}
\mathcal{F}(\mathcal{U}, P) &= \sum_{i=1}^N \left(\prod_{j=1}^{i-1} (1-p_j)\right) p_i \mathsf{F}(P/k, \rho_i) + \left(\prod_{i=1}^{N} (1-p_i)\right) \mathsf{F}(P/k, I/d) \nonumber \\
                            &\geq \sum_{i=1}^N \left(\prod_{j=1}^{i-1} (1-p_j)\right) p_i \left|\braket{\phi_i | \tilde{\phi_i}}\right|^2 \label{eq:fidelitybound}
\end{align}
where items \ref{item:pure-f}, \ref{item:unitary-inv} and \ref{item:partial} of Fact \ref{fac:distanceproperties} were used for the inequality.

We will show the following theorem, which proves Theorem \ref{thm:rsp-prot-2}.
\begin{theorem} \label{prop:rejectionsampling}
For every $\varepsilon_a > 0$, such a choice of $\mathcal{U}$ exists if we take $\eta = \frac{\eps_a^2}{2}$ and $N=2\frac{d \ln(1/\varepsilon_a)}{\eps_a^2k}+\Theta(1)$. 
\end{theorem}

 In order to prove Theorem \ref{prop:rejectionsampling}, we will need two preliminary lemmas. We will actually only need the versions of the lemmas with pure states, but we state the mixed state cases for generality.
\begin{lemma} \label{lem:damagecontrol}
    Let $\rho$ and $\sigma$ be two states, and let $P$ be a random projector (according to some distribution, not necessarily Haar-random in $G(d,k)$) with $\mathbb{E}[P] = \gamma I$. For $c \in (0,1)$, let:
    \begin{eqnarray*}
    M & = & I - cP, \\
    \beta & = & \tr(M^2 \rho), \\
    \rho' & = & \frac{1}{\beta} M \rho M.
    \end{eqnarray*}
    Then,
    \[\mathbb{E}[\beta \cdot\F(\rho', \sigma)] \geq \mathbb{E}[\beta] \cdot\F(\rho, \sigma) - c^2 \gamma.\]
\end{lemma}

\begin{proof}
By Uhlmann's theorem, we can find purifications $\ket{\psi}$ and $\ket{\phi}$ of $\rho$ and $\sigma$, respectively, with $\braket{\psi|\phi}$ real and with $\braket{\psi | \phi}^2 = \F(\rho, \sigma)$. Let
\begin{eqnarray*}
Q & = & P \otimes I, \\
\alpha & = & \tr(P \rho) = \braket{\psi|Q|\psi},
\end{eqnarray*}
so that
\[\beta =  \Tr(M^2 \rho) = \braket{\psi | (M \otimes I) ^2 | \psi}  = 1+ (c^2-2c)\alpha.  \]
Note that $\mathbb{E}[\alpha] = \Tr(\mathbb{E}[P]\rho) = \gamma\Tr(\rho)=\gamma$. 

Now, letting
\[\ket{\psi'} = \frac{1}{\sqrt{\beta}} (M \otimes I) \ket{\psi},\]
we compute:
\begin{align*}
\beta\braket{\psi|\phi}^2 - \beta\left|\braket{\psi'|\phi}\right|^2  &= \beta \braket{\psi|\phi}^2 - \left|\braket{\psi | I-cQ | \phi}\right|^2\\
&= \left(1+(c^2-2c)\alpha\right)\braket{\psi|\phi}^2 - \braket{\psi|\phi}^2 + 2c \braket{\psi|\phi}\text{Re}(\braket{\psi|Q|\phi}) - c^2\left|\braket{\psi|Q|\phi}\right|^2\\
&= 2c \braket{\psi | \phi} (\text{Re}(\braket{\psi|Q|\phi}) - \alpha \braket{\psi | \phi} )+ c^2 \left(\alpha\braket{\psi|\phi}^2 - \left|\braket{\psi|Q|\phi}\right|^2 \right) \\
& \leq 2c \braket{\psi | \phi} (\text{Re}(\braket{\psi|Q|\phi}) - \alpha \braket{\psi | \phi} )+ c^2\alpha\braket{\psi|\phi}^2
\end{align*}
Hence,
\begin{align*}
\mathbb{E}[\beta]\cdot\left|\braket{\psi|\phi}\right|^2 - \mathbb{E}[\beta\cdot\left|\braket{\psi'|\phi}\right|^2] &= \mathbb{E}[\beta\cdot\left|\braket{\psi|\phi}\right|^2 - \beta\cdot\left|\braket{\psi'|\phi}\right|^2]\\
&\leq 2c \braket{\psi | \phi} \text{Re}(\braket{\psi|\mathbb{E}[Q - \alpha I]|\phi}) + c^2 \mathbb{E}[\alpha]\braket{\psi|\phi}^2 \\
&= 2c \braket{\psi | \phi} \text{Re}(\braket{\psi|\gamma I - \gamma I|\phi}) + c^2 \gamma \braket{\psi|\phi}^2\\
&\leq c^2 \gamma
\end{align*}
where the fact that $\braket{\psi|\phi}$ is real was used for the second equality and the Cauchy-Schwarz inequality was used for the second inequality. Hence,
\begin{align*}
\mathbb{E}[\beta] \F(\rho,\sigma) - c^2 \gamma &= 
\mathbb{E}[\beta] \left|\braket{\psi|\phi}\right|^2 - c^2 \gamma \\
& \leq \mathbb{E}[\beta\cdot\left|\braket{\psi'|\phi}\right|^2] \\
& \leq \mathbb{E}[\beta\cdot\F(\rho', \sigma)] \\
\end{align*}
where we used the fact that $\psi'$ is a purification of $\rho'$ in the last line. This completes the proof. 
\end{proof}

\begin{lemma} \label{lem:postselection}
Let $\rho$ and $\sigma$ be arbitrary states, and let $P$ be a random projector with $\mathbb{E}[P] = \gamma I$. For $\eta \in (0,1]$, consider the random variables:
\begin{eqnarray*}
p & = & \eta \tr(P \rho) \\
\rho' & = & \frac{\eta}{p} P \rho P \\
q & = & \eta \tr(P \sigma) \\
\sigma' & = & \frac{\eta}{q} P \sigma P.
\end{eqnarray*}
Then
\[\mathbb{E}[p\cdot \F(\rho', \sigma')] \geq \eta \gamma \cdot\F(\rho,\sigma) \]
\end{lemma}
\begin{proof}
Let $\ket{\psi}, \ket{\phi}$ be purifications of $\rho$ and $\sigma$, respectively, with $\braket{\psi | \phi}$ real and with $\braket{\psi|\phi}^2 = \F(\rho,\sigma)$. 
Letting $Q = P \otimes I$, and
\[\ket{\psi} = \alpha \ket{\phi} + \beta \ket{\phi^\bot} \]
with $\beta$ real, we compute:
\begin{equation}\label{eq:p}
p = \eta\braket{\psi | Q | \psi} = \alpha^2q  + \eta\beta^2 \braket{\phi^\bot | Q | \phi^\bot} + 2 \eta\alpha \beta   \text{ Re}(\braket{\phi | Q | \phi^\bot}).
\end{equation}
Moreover,
\begin{equation}\label{eq:eta-Q}
\eta^2\left|\braket{\psi|Q|\phi}\right|^2 = \left|\alpha q  + \eta\beta \braket{\phi | Q | \phi^\bot}\right|^2 = \alpha^2 q^2 + 2\eta\alpha \beta q \text{ Re}(\braket{\phi | Q | \phi^\bot}) +  \eta^2\beta^2 \left|\braket{\phi | Q | \phi^\bot}\right|^2.
\end{equation}
Now setting
\[\ket{\psi'} = \sqrt{\frac{\eta}{p}} Q\ket{\psi}\]
\[\ket{\phi'} = \sqrt{\frac{\eta}{q}} Q\ket{\phi}\]
and using equations \eqref{eq:p} and \eqref{eq:eta-Q},
\begin{align*}
    p- p\left|\braket{\psi'|\phi'}\right|^2 &= p - \frac{\eta^2}{q} |\braket{\psi|Q|\phi}|^2\\
    & = \alpha^2q  + \eta\beta^2 \braket{\phi^\bot | Q | \phi^\bot} + 2 \eta\alpha \beta\text{ Re}(\braket{\phi | Q | \phi^\bot}) \\
    & \quad - \alpha^2 q - 2\eta\alpha \beta \text{ Re}(\braket{\phi | Q | \phi^\bot}) - \frac{\eta^2\beta^2}{q} |\braket{\phi | Q | \phi^\bot}|^2\\
    &= \eta\beta^2\left(\braket{\phi^\bot | Q | \phi^\bot} - \frac{\eta}{q}\left|\braket{\phi | Q | \phi^\bot}\right|^2\right)\\
    &\leq \eta\beta^2\braket{\phi^\bot | Q | \phi^\bot} \\
    & = \eta(1-\braket{\psi|\phi}^2)\cdot \braket{\phi^\bot | Q | \phi^\bot}
\end{align*}
where in the last line we have used the fact that $\alpha=\braket{\psi|\phi}$, and $\alpha^2+\beta^2=1$.

Now, note that $\mathbb{E}[p]=\eta\gamma$. Hence,
\begin{align*}
 \eta\gamma - \mathbb{E}[p\cdot\left|\braket{\psi'|\phi'}\right|^2] &= \mathbb{E}[p- p\cdot\left|\braket{\psi'|\phi'}\right|^2] \\
 & \leq \eta(1-\braket{\psi|\phi}^2)\cdot[\braket{\phi^\bot | \mathbb{E}[Q] | \phi^\bot}]\\
 & = \eta(1 - \braket{\psi | \phi}^2)\cdot[\braket{\phi^\bot | \gamma I | \phi^\bot}] \\
 & = \eta\gamma(1 - \braket{\psi | \phi}^2).
\end{align*}
This implies:
\[\eta\gamma \F(\rho,\sigma) = \eta\gamma \braket{\psi | \phi}^2 \leq \mathbb{E}[p\left|\braket{\psi'|\phi'}\right|^2] \leq \mathbb{E}[p\cdot\F(\rho', \sigma')],  \]
since $\psi'$ and $\phi'$ are purifications of $\rho'$ and $\sigma$, respectively. 
\end{proof}

We can now finish the correctness proof for our protocol.
\begin{proof}[Proof of Theorem \ref{prop:rejectionsampling}]
Let $\mathcal{U}=\{U_1\ldots U_N\}$ be $N$ i.i.d. samples from any exact $1$-design on $U(d)$ (for example, the Pauli group forms such a design when $d$ is a power of $2$). By the definition of a $1$-design, for any $P$, $E_{U_i}[U_iPU^\dagger_i] = \frac{k}{d}I$.

It is enough to prove that with the choices of $\eta$ and $N$ in the theorem statement,
\[\mathbb{E}_{\mathcal{U}, P}[\mathcal{F}(\mathcal{U},  P)] \geq 1 - \eps_a^2.\]
By the probabilistic method, we can then select $\mathcal{U}$ for which the $\mathbb{E}_P[\mathcal{F}(\mathcal{U},P)] \geq  1 - \eps_a^2$. We now turn to lower bounding this expectation. In fact, we can lower bound $\mathbb{E}_{\mathcal{U}}[\mathcal{F}(\mathcal{U},  P)]$ for any fixed $P$.
Let $\delta= \frac{\eta k}{d}$. Note that $\mathbb{E}[p_i]=\delta$ for all $i$ because $\mathbb{E}[\eta Q_i^\dagger Q_i] = \frac{\eta k}{d} I$. We will now lower bound the expected value of a given term in the expression \eqref{eq:fidelitybound}. The projectors $U_i\overline{P}U_i^\dagger$ satisfy the conditions of Lemmas \ref{lem:damagecontrol} and \ref{lem:postselection} with $\gamma=\frac{k}{d}$, and $\beta = 1-\delta$. Therefore, by applying Lemma \ref{lem:postselection}, we have:
\begin{align*}
\mathop{\mathbb{E}}_{U_1\ldots U_i}\left[\left(\prod_{j=1}^{i-1} (1-p_j)\right) p_i \left|\braket{\phi_i| \tilde{\phi_i}}\right|^2\right] &= \mathop{\mathbb{E}}_{U_1\ldots U_{i-1}}\left[\mathop{\mathbb{E}}_{U_i}\left[\left(\prod_{j=1}^{i-1} (1-p_j)\right) p_i \left|\braket{\phi_i | \tilde{\phi_i}}\right|^2 \middle| U_1,...,U_{i-1}\right]\right]\\
&= \mathop{\mathbb{E}}_{U_1\ldots U_{i-1}}\left[ \left(\prod_{j=1}^{i-1} (1-p_j)\right) \mathop{\mathbb{E}}_{U_i}\left[ p_i \left|\braket{\phi_i | \tilde{\phi_i}}\right|^2 \middle| U_1,...,U_{i-1}\right]\right]\\
& \geq \mathbb{E}_{U_1\ldots U_{i-1}}\left[ \left(\prod_{j=1}^{i-1} (1-p_j)\right) \eta\cdot\frac{k}{d}\left|\braket{\psi_{i-1} | \Phi_+}\right|^2\right] \\
& = \delta \cdot\mathbb{E}\left[ \left(\prod_{j=1}^{i-1} (1-p_j)\right) |\braket{\psi_{i-1} | \Phi_+}|^2\right].
\end{align*}
Also, note that $\sqrt{I-\eta Q_i} = I - c Q_i$ where $c = 1 - \sqrt{1-\eta}$, as can be verified by squaring the right-hand-side. Applying Lemma \ref{lem:damagecontrol}, we have, for all $j$:
\[\mathop{\mathbb{E}}_{U_j}[(1-p_j) \left[\left|\braket{\psi_j | \Phi_+}\right|^2 \middle| U_1,...U_{j-1}\right] \geq (1-\delta)|\braket{\psi_{j-1} | \Phi_+}|^2 - \frac{c^2 k}{d}.\]
Applying the above repeatedly and using that $\ket{\psi_0} = \ket{\Phi_+}$, we get:
\begin{align*}
\mathop{\mathbb{E}}_{U_1\ldots U_i}\left[\left(\prod_{j=1}^{i-1} (1-p_j)\right) p_i |\braket{\phi_i | \tilde{\phi_i}}|^2\right] &\geq \delta(1-\delta)^{i-1}\left(1 - (i-1) \frac{c^2 \delta}{\eta(1-\delta)}\right).
\end{align*}

Substituting the above into equation \eqref{eq:fidelitybound} yields:
\begin{align}
\mathbb{E}_{\mathcal{U}}[\mathcal{F}(\mathcal{U},  P)] &\geq \sum_{i=1}^N \delta(1-\delta)^{i-1} - \frac{c^2\delta^2}{\eta (1-\delta)} \sum_{i=1}^N (i-1) (1-\delta)^{i-1} \nonumber \\
& \geq \delta\sum_{i=1}^N(1-\delta)^{i-1} - \frac{c^2\delta^2}{\eta (1-\delta)} \sum_{i=0}^\infty i (1-\delta)^i. \label{eq:series}
 \end{align}
The geometric series $\sum_{i=1}^N(1-\delta)^{i-1}$ is equal to $\dfrac{1-(1-\delta)^N}{\delta}$. The series $\sum_{i=0}^\infty i(1-\delta)^{i-1}$ can be computed by differentiating the geometric series $\sum_{i=0}^\infty (1-\delta)^i$, and its value is $\dfrac{1-\delta}{\delta^2}$. Moreover,
\begin{align*}
c = 1 - \sqrt{1-\eta} = \frac{\eta}{1+\sqrt{1-\eta}} \leq \eta.
\end{align*}
Putting these into \eqref{eq:series} we get,
\begin{align*}
\mathbb{E}_{\mathcal{U}}[\mathcal{F}(\mathcal{U},  P)] & \geq  \delta\cdot\frac{1 - (1-\delta)^N}{\delta} - \frac{c^2\delta^2}{\eta(1-\delta)}\cdot\frac{1-\delta}{\delta^2} \\
& \geq  1 - (1-\delta)^N - \frac{c^2}{\eta} \\
& \geq 1 - (1-\delta)^N - \eta.
\end{align*}
Therefore, taking:
\begin{eqnarray*}
\eta & = & \frac{\eps_a^2}{2} \\
N & = & \left\lceil \frac{\ln(2/\eps_a^2)}{\delta} \right \rceil = \left\lceil \frac{2 d \ln(2/\eps_a^2)}{\eps_a^2 k} \right \rceil
\end{eqnarray*}
makes the expectation at least $1- \eps_a^2$, as desired.
\end{proof}

\section{Resource lower bounds}

In the previous section, we described protocols which achieved constant-error RSP of projectors in $G(d,k)$ with $\approx \log \frac{d}{k}$ bits of communication from Alice to Bob and $\approx \log d$ EPR pairs. We now turn to showing that these are near-optimal. While proving the communication lower bound is simple, proving the entanglement lower bound is a lot more involved and takes up the bulk of the section.

We note that no lower bounds can be shown for protocols with $\eps_r \geq 1 - \frac{k}{d}$, as this can be achieved with no communication or entanglement at all by having Bob output the maximally mixed state. Provided that $\eps_r$ is smaller than this, we can prove a communication lower bound, which will be a simple consequence of the following result.
\begin{fact}[See e.g. Theorem 4.1 of \cite{nosignalling}] \label{thm:nonsignalling}
For $n \geq 1$: consider the following task. Alice is given a uniformly random string $x \in \{0,1\}^n$. Alice and Bob, who initially start out with an arbitrary maximally entangled state and have access to arbitrary shared randomness, run a two-way classical communication protocol. At the end of the protocol, Bob produces a string $x' \in \{0,1\}^n$. If $m$ is an upper bound on the number of bits that were sent from Alice to Bob during the course of the protocol, we have
\[\Pr[x=x'] \leq 2^{m-n}\]
\end{fact}

We can then show our lower bound for the communication cost of $(d,k)$-RSP protocols.
\begin{theorem}\label{thm:comm-lowerbound}
    A $(d,k)$-RSP protocol with communication $m$ and relaxed average-case error $\varepsilon_r$ must satisfy
    \[m \geq \left\lfloor\log \frac{d}{k} \right\rfloor + \log (1-\varepsilon_r).\]
\end{theorem}
\begin{proof}
For $r = \lfloor{\frac{d}{k}} \rfloor$, let $P_1, \ldots, P_r$ be mutually orthogonal elements of $G(d,k)$. Complete them to a projective measurement with an additional projector $\tilde{P}$. Given the protocol in the statement of the theorem, consider the following protocol for transmitting $t \in [r]$ from Alice to Bob. Using shared randomness, they sample a uniformly random unitary $U \in U(d)$; they run the original protocol with Alice input $U P_t U^\dagger$, and at the end of the protocol, Bob measures $B_2$ according to the projective measurement $U P_1 U^\dagger, \ldots, U P_r U^\dagger, U \tilde{P} U^\dagger$, getting outcome $t' \in [r]$ (if the last outcome is obtained, Bob picks $t'$ uniformly at random). We see that
\[1-\eps_r \leq \Pr[t=t'] \leq 2^{m - \lfloor \log \frac{d}{k}\rfloor}\]
Where $t$ was encoded as a binary string and Fact \ref{thm:nonsignalling} was applied. This gives the statement of the theorem. 
\end{proof}
Note that this communication lower bound straightforwardly works even for RSP protocols with two-way communication and any amount of shared entanglement and randomness, and it lower bounds not only for the total communication of such a protocol, but the communication from Alice to Bob specifically.

Now we move on to proving a lower bound on the entanglement cost of $(d,k)$-RSP. In this regard, in addition to the condition $\eps_r$ not be too close to $1$, lower bounds can only be shown under the additional condition that $m = o(d)$. This is because, as noted in the introduction, there exists a good $(d,k)$-RSP protocol for relaxed error, namely the $\eps$-net protocol, which is completely classical (i.e. requires no entanglement) and which involves only $\Theta(d)$ bits of communication. As it turns out, we can show that any protocol for which the bounds on both $\eps_r$ and communication are true, and which uses pure entanglement, must essentially be using $\log d - O(1)$ EPR pairs.
\begin{theorem} \label{thm:minentropy}
    For all $\gamma > 0$, any $(d,k)$-RSP protocol with relaxed average error $\eps_r$, $m$ bits of communication, and initial pure shared state $\ket{\sigma}^{AB}$, must satisfy
    \[H^{\delta + \gamma}_{\text{min}}(A)_\sigma \geq \log d - 3 \log(1/\gamma) - O(1) \]
    where $\delta = F\left(\frac{k}{d}+O(\sqrt{\frac{m}{d}}), 1-\varepsilon_r\right) $, where $F$ is the truncated fidelity function.
\end{theorem}
In view of Fact \ref{fc:distillation-lb}, Theorem \ref{thm:minentropy} implies that the initial entangled state is such that almost $\log d$ many EPR pairs can be extracted from it via LOCC, with a failure probability that is dependent on the smoothing. In particular, nothing is implied about the entanglement if $m = \Theta(d)$, in which case $\delta = 1$, and the failure probability of the distillation protocol in Fact \ref{fc:distillation-lb} is $1$.

\subsection{The proof strategy}

Our proof strategy for Theorem \ref{thm:minentropy} will run as follows. We fix a particular $(d,k)$-RSP protocol. Let $p$ be the distribution on $G(d,k)\times \{0,1\}^m$ obtained by sampling a $P$ Haar-randomly from $G(d,k)$ and then sampling $c$ according to the message distribution induced by the protocol on input $P$. If Alice sends message $c$, the pure state shared by Alice and Bob at the end of the protocol is $\ket{\chi_c}^{IA'B_1B_2}$.  For each $c \in \{0,1\}^m$, define the state $\rho_c^{B_1B_2}$ by:
\[\rho^{B_1B_2}_c = \mathbb{E}_{P \sim p(\cdot \mid c)}[\chi^{B_1B_2}_{c, P}].\]
Moreover, let $\lambda = \{\lambda_j\}_j$ denote the Schmidt spectrum of $\ket{\sigma}^{AB}$, and $\nu_c = \{\nu_{c,j}\}_j$ denote the eigenspectra of the $\rho^{B_1B_2}_c$. We start by proving the following lemma.
\begin{lemma} \label{prop:majorizes}
The spectra $\lambda$ and $\nu_c$ satisfy 
\[\lambda \preceq \sum_c p(c) \nu_c.\]
\end{lemma}

To facilitate proving Lemma \ref{prop:majorizes}, we will consider running the protocol in superposition over a number of Alice inputs $P_1, \ldots , P_N$ to be determined later (along with the number $N$).\footnote{Ideally, we would run a uniform superposition over all $P\in G(d,k)$, but this will make the overall state infinite-dimensional, which is unpleasant to analyze. Although it would appear that results in the same vein as Fact \ref{thm:jonathanplenio} are known to hold in infinite dimensions (\cite{vanluijk2024purestateentanglementvon}), it seems safer to instead pick a finite set of $P_i$ such that the superposition over them approximates the superposition over $G(d,k)$ for our purposes.} In order to run the protocol in superposition,  we introduce a new dimension-$N$ register $I$ on Alice's side, initially in the uniform superposition. The original protocol is then carried out with Alice's first operation controlled on $I$, i.e.  the isometry $U: IA \mapsto IA'C$ Alice applies is given by
\[U = \sum_i \ket{i} \bra{i}^I \otimes U_{P_i}^{A}\]
where $U_{P_i}^A$ is the isometry Alice applies on input $P_i$ in the original protocol. The rest of Alice and Bob's operations in the superposition protocol are the same as in the original protocol. We will invoke Fact \ref{thm:jonathanplenio} to infer properties of the spectrum of the initial state $\ket{\sigma}^{AB}$ from the spectra of the final states of this superposition protocol.

Let $\tilde{p}$ be the distribution on $[N] \times \{0,1\}^m$ that is induced by first sampling $i \in [N]$ and then sampling $c$ according to the marginal distribution $p(c|P_i)$. In the superposition protocol, the final state of Alice and Bob's registers for a given outcome $c$ is given by
\[\ket{\chi_c}^{IA'B_1B_2} = \sum_i \sqrt{\tilde{p}(i| c)} \ket{i}^I \ket{\chi_{P_i,c}}^{A'B_1B_2}.\]
The marginal state on the registers $B_1B_2$ for outcome $c$ is then,
\[ \chi^{B_1B_2}_c =\sum_i \tilde{p}(i|c)\chi^{B_1B_2}_{P_i,c}.\]

With this notation, we are now ready to prove Lemma \ref{prop:majorizes}.
\begin{proof}[Proof of Lemma \ref{prop:majorizes}]
For any $\eps > 0$, we will show that for all $l$, 
\[\sum_{i=1}^l \lambda_i \leq \left(\sum_c p(c) \sum_{i=1}^l \nu_{c,i}\right) + \eps.\]
Since $\eps$ could have been taken to be arbitrarily small, this establishes the lemma.

Let $\tilde{\nu}_c$ denote the eigenspectra of $\chi^{B_1B_2}_c$ (which is equal to the Schmidt spectra of $\ket{\chi}^{IA'B_1B_2}_c$). From Fact \ref{thm:jonathanplenio}, since the ensemble $\{(\tilde{p}(c), \ket{\chi_c}^{IA'B_1B_2})\}_{c \in \{0,1\}^m}$ can be produced probabilistically from the initial state $\ket{\sigma}^{AB}$, we always have that
\[\sum_{i=1}^l \lambda_i \leq \sum_c \tilde{p}(c) \sum_{i=1}^l \tilde{\nu}_{c,i} = \sum_c p(c) \sum_{i=1}^l \frac{\tilde{p}(c)}{p(c)}\tilde{\nu}_{c,i}.\]

We will show that, for some sufficiently large value of $N$, it is possible to choose $P_1, \ldots, P_N$ so that, for all $c$, 
\begin{equation} \label{eq:trace}
\left\|\rho^{B_1B_2}_c - \frac{\tilde{p}(c)}{p(c)}\chi^{B_1B_2}_c\right\|_\tr \leq \eps.
\end{equation}
Now the sum of the largest $l$ eigenvalues of a density matrix $\omega$ can be written as $\sup\{\Tr(M\omega): 0 \preceq M \preceq I, \mathrm{rank}(M)=l\}$. Moreover, the optimal $M$ can be taken to be a projector, which means that due to point 3 of Fact \ref{fac:distanceproperties}, the optimization takes values that are at most $\eps$-far for matrices that are $\eps$-far in trace distance. Therefore, the trace distance condition will imply for all $l$ and for all $c$, 
\[\left|\sum_{i=1}^l \nu_{c, i} - \frac{\tilde{p}(c)}{p(c)} \tilde{\nu}_{c,i}\right| \leq \eps,\]
and the statement follows. We pick $N$ and $P_1, \ldots, P_N$ such that equation \eqref{eq:trace} holds probabilistically. We will pick the $P_i$ independently and uniformly at random in $G(d,k)$ and bound the probability that this equation does not hold for some $c$. We have:
\begin{align*}
\frac{\tilde{p}(c)}{p(c)}\chi^{B_1B_2}_c = \frac{1}{N} \sum_{i=1}^N \frac{p(c|P_i)}{p(c)}\rho^{B_1B_2}_{P_i,c}
\end{align*}
just by expanding out the definition of $\tilde{p}(c)$. It is not hard to see that
\[\mathbb{E}_{P \sim p} \left[\frac{p(c|P)}{p(c)} \rho^{B_1B_2}_{P,c} \right] = \mathbb{E}_{P \sim p} \left[\frac{p(P|c)}{p(P)} \rho^{B_1B_2}_{P,c} \right] = \mathbb{E}_{P \sim p(\cdot |c)} \left[ \rho_{P,c} \right] = \rho^{B_1B_2}_c.\]
Since $\frac{\tilde{p}(c)}{p(c)}\chi^{B_1B_2}_c$ is a sample mean of independent $P_i$, by applying Fact \ref{fact:weaklaw} component-wise on each matrix entry, we get that for every $c$, there exists $N_c$ such that, for all $N \geq N_c$, the probability that equation \eqref{eq:trace} doesn't hold is strictly less than $\frac{1}{2^m}$. Setting $N = \max_c N_c$, and taking a union bound, this shows that there exists a choice of $P_1, \ldots, P_N$ such that equation \eqref{eq:trace} holds for all $c$. This completes the proof.
\end{proof}

In view of this lemma, we can prove Theorem \ref{thm:minentropy} assuming the following technical result which the rest of the section will be devoted to proving.
\begin{theorem} \label{thm:eigvalbound}
Let $\mu$ be a probability measure on $G(d,k)$ with infinity norm $K$, i.e., for all measurable sets $A \subseteq G(d,k)$, $\mu(A) \leq K\cdot \mu_{\text{Haar}}(A)$, where $\mu_{\text{Haar}}$ is the Haar measure on $G(d,k)$. Suppose there is a measurable map $P\to \omega_P$ where $\omega_P$ is a density matrix on registers $B_1B_2$. Take:
\[\omega = \mathbb{E}_{P \sim \mu}[\omega_P]\]
\[\varepsilon = \mathbb{E}_{P\sim \mu}\left[\Tr(\omega_P (I^{B_1} \otimes (I-P)^{B_2}))\right].\]
Then, letting $\nu_1 \geq \nu_2 \geq \ldots $ be the spectrum of $\omega$, there is an absolute constant $A > 0$ such that for all $l$, 
\[\sum_{j=1}^l \nu_j \leq F\left(\frac{k}{d} + \frac{A}{\sqrt{d}}\left( \sqrt{\log K} + \sqrt{l}\right), 1 - \varepsilon\right).\]
\end{theorem}

With this result, we can now prove Theorem \ref{thm:minentropy}.

\begin{proof}[Proof of Theorem~\ref{thm:minentropy}]
Let $\eps_c = \mathbb{E}_{P \sim p(\cdot|c)}\left[\Tr\left(\chi^{B_1B_2}_{c,P}(I^{B_1}\otimes(I-P)^{B_2})\right)\right]$. Then $\mathbb{E}_{c \sim p}\eps_c = \eps_r$. We will apply Theorem \ref{thm:eigvalbound} to each of the matrices $\chi^{B_1B_2}_c$.

From Lemma \ref{prop:majorizes}, we have that, for all $l$, 
\begin{align*}
\sum_{j=1}^l \lambda_j &\leq \sum_c p(c) \sum_{j=1}^l \nu_{c,j} \\
                        &\leq \sum_c p(c) F\left(\frac{k}{d} + A\left( \sqrt{\frac{\log(1/p(c))}{d}} + \sqrt{\frac{l}{d}}\right), \; 1 - \varepsilon_c\right) \\
                        &\leq F\left(\frac{k}{d} + A\sqrt{\frac{l}{d}} + \frac{A}{\sqrt{d}} \sum_c p(c) \sqrt{\log(1/p(c))}, \; 1 - \varepsilon_r \right) \\
                        & \leq F\left(\frac{k}{d} + A\sqrt{\frac{l}{d}} + \frac{A}{\sqrt{d}} \cdot\sqrt{\sum_c p(c) \log(1/p(c))}, \; 1 - \varepsilon_r \right) \\
                        &\leq F\left(\frac{k}{d} + A \left(\sqrt{\frac{m}{d}} + \sqrt{\frac{l}{d}}\right), \; 1 - \varepsilon_r \right)
\end{align*}
where Theorem \ref{thm:eigvalbound} was applied for the first inequality, Jensen's inequality yields the second and third inequalities, also noting that $F$ is monotonically increasing in the first argument for the third inequality. The last inequality is obtained by noting that $\sum_cp(c)\log(1/p(c))$ is equal to the entropy of $p(c)$, which is at most $m$, since $c$ is a bit string of length $m$. Hence, for any $S > 0$ and for all $l$, 
\begin{align*}
\sum_{i=1}^l (\lambda_i - S) &\leq F\left(\frac{k}{d} + A\sqrt{\frac{m}{d}} + \sqrt{\frac{A^2l}{d}}, \; 1 - \varepsilon_r \right) - \frac{Sd}{A^2} \left(\sqrt{\frac{A^2l}{d}}\right)^2 \\
                                &\leq  F\left(\frac{k}{d} + A \sqrt{\frac{m}{d}}, \; 1 - \varepsilon_r \right) + O\left(\left(\frac{Sd}{A^2}\right)^{-1/3}\right)
\end{align*}
where Lemma \ref{lem:boundfidelity} was applied for the last inequality. Setting $S = \Theta\left(\frac{1}{\gamma^3 d}\right)$ and using the characterization of the smooth min-entropy in Lemma \ref{lem:minentropydef} completes the proof.
\end{proof}

\subsection{The key semidefinite program}
We first prove a result about the optimal value of a semidefinite program, which will be needed in our proof of Theorem \ref{thm:eigvalbound}. Taking $M$ and $N$ to be two nonzero projectors on $\mathbb{C}^d$, and taking $t \in (0,1)$, we will study the following semidefinite program:
\begin{equation}\label{eq:sdp}
\begin{aligned}
\max_{\rho} \quad & \Tr(M \rho) \\
\text{s.t.} \quad 
& \Tr(N \rho) \ge t \\
& \rho \succeq 0, \Tr(\rho)=1. \\
\end{aligned}
\end{equation}\label{eq:main}
This SDP is useful to us because the quantity we are trying to upper bound in Theorem \ref{thm:eigvalbound} is precisely the average version of such an SDP. Specifically, in the setting of Theorem \ref{thm:eigvalbound}, we know that $\mathbb{E}_P[\Tr((I^{B_1}\otimes P^{B_2})\omega_P)] \geq 1-\eps$, and we are trying to upper bound the sum of the first $l$ eigenvalues of $\omega_P$ on average. But the sum of the first $l$ eigenvalues of $\omega_P$ can be expressed as the maximum value of $\Tr(Q\omega_P)$, for a rank-$l$ projector $Q$, and we want to know what the maximum value of  $\mathbb{E}_P[\Tr(Q\omega_P)]$ can be when $\omega_P$-s are subject to the $\mathbb{E}_P[\Tr((I^{B_1}\otimes P^{B_2})\omega_P)] \geq 1-\eps$.

We will prove the following lemma about the value of the program \eqref{eq:sdp}.
\begin{lemma} \label{prop:optsdp}
The optimal value of program \eqref{eq:sdp} is $F(\|MNM\|_\infty, t)$, where $F$ is the truncated fidelity function.
\end{lemma}
By the previous argument, this lets us recast upper bounding the sum of the first $l$ eigenvalues of the average $\mathbb{E}_P[\omega_P]$, as simply upper bounding $\mathbb{E}_P [F(\|PQP\|_\infty,1-\eps)]$ for a rank-$l$ projector $Q$, which will be much more convenient. For details of this argument, see  Section \ref{sec:finproof}. The rest of this section will be dedicated to proving Lemma \ref{prop:optsdp}.

We start by noting that this program admits a pure optimal solution in view of the following simple lemma.\footnote{The natural extension of this result to three matrices is false because the expectation values of a single-qubit mixed state under the three non-identity Pauli matrices completely characterize it.}
\begin{lemma} \label{lem:pure}
Let $A$ and $B$ be Hermitian matrices and let $\rho$ be a mixed state. There exists a pure state $\ket{\psi}$ such that:
\[\Tr(A \rho) = \braket{\psi|A|\psi} \]
\[\Tr(B \rho) = \braket{\psi|B|\psi}. \]
\end{lemma}
\begin{proof}
Consider the spectral decomposition of $\rho$:
\[\rho = \sum_i p_i \ket{\psi_i}\bra{\psi_i}.\]
Let $C = A + iB$ and set $a = \Tr(A \rho)$, $b = \Tr(B \rho)$, both of which are real because $A$ and $B$ are assumed to be Hermitian. Then,
\[a + ib = \Tr(\rho C) = \sum_i p_i \braket{\psi_i | C | \psi_i}\]
Now the set of numerical values of $\bra{\phi}C\ket{\phi}$ is a convex set,\footnote{The most general form of this result is known as the Toeplitz-Hausdorff Theorem.} and therefore there exists $\ket{\psi}$ with $\Tr(\rho C) = \braket{\psi | C | \psi}$. The result follows from comparing the real and imaginary parts.
\end{proof}

Next, we prove a special case of Lemma \ref{prop:optsdp}.
\begin{lemma} \label{prop:optsdpspecial}
    The statement of Lemma \ref{prop:optsdp} holds when $M$ and $N$ are rank-1 projectors on $\mathbb{C}^2$.
\end{lemma}
\begin{proof}
Assume without loss of generality that $N = \state{0}$ and $M = \state{\phi}$, with $\ket{\phi} = \alpha \ket{0} + \beta \ket{1}$ and $\alpha$ real and nonnegative. We then have that $\|MNM\|_\infty = \alpha^2$. In view of Lemma \ref{lem:pure}, the SDP reduces to optimizing over a pure state $\ket{\psi} = \gamma\ket{0}+\delta\ket{1}$, where again $\gamma$ can be assumed to be real and nonnegative without loss of generality. The objective function is then $|\alpha \gamma + \beta^* \delta|^2$, and the constraints are that $\gamma^2 + |\delta|^2 = 1$ and $\gamma^2 \geq t$.

In the event that $\alpha^2 \geq t$, setting $\gamma = \alpha$, $\delta = \beta$ produces a value of 1, which is the highest possible. If not, $|\alpha\gamma +\beta^*\delta|^2$ is maximized when the two summands have the same phase. Since $\alpha$ and $\gamma$ are real, an optimal $\delta$ there makes $\beta^* \delta$ real and nonnegative. Letting $\beta = \sqrt{1-\alpha^2} \exp(i \theta)$, $\delta$ then should be $\sqrt{1-\gamma^2}\exp(i\theta)$, and the objective function becomes $(\alpha\gamma + \sqrt{(1-\alpha^2)(1-\gamma^2)})^2$, with the constraint $\gamma \geq t$. In the $\alpha \leq t$ case, it is then a simple calculus exercise to see that it is optimal to set the single remaining variable $\gamma = \sqrt{t}$, resulting in the promised objective. 
\end{proof}

We will now handle the general case of Lemma \ref{prop:optsdp}. In order to do this, we will need the following result.
\begin{fact}[Jordan's Lemma, see \cite{Zhu_2013}] \label{lem:jordan}
Let $U$ and $V$ be two subspaces of $\mathbb{C}^d$, of dimension $d_1$ and $d_2$ respectively, and let $\Pi_U$ and $\Pi_V$ be the projectors onto them. There exist orthonormal bases $u_1, \ldots, u_{d_1}$ and $v_1, \ldots, v_{d_2}$ of $U$ and $V$, respectively, such that $\braket{u_i | v_j} = 0$ for all $i \neq j$. We have that the values of $\left|\braket{u_i | v_i}\right|$ for $1 \leq i \leq \min(d_1, d_2)$ are the singular values of $\Pi_U \Pi_V$, and in particular that
\[\max_{i} \left|\braket{u_i | v_i}\right|^2 = \|\Pi_U \Pi_V \Pi_U\|_\infty = \|\Pi_V \Pi_U \Pi_V\|_\infty.\]
The bases $\{u_i\}_i, \{v_i\}_i$ are called the Jordan bases of the subspaces $U$ and $V$.
\end{fact}

\begin{proof}[Proof of Lemma \ref{prop:optsdp}]
The dual of the SDP \eqref{eq:sdp} is:
\begin{equation}\label{eq:dual}
\begin{aligned}
\min_{\lambda,\mu}\quad & -\lambda\, t - \mu \\
\text{s.t.}\quad
&  \lambda N + \mu I - M \succeq 0, \\
& \lambda \ge 0, \quad \mu\in\mathbb{R}.
\end{aligned}
\end{equation}
It isn't hard to check that program \eqref{eq:sdp} satisfies Slater's condition. Therefore, strong duality holds, which implies that the values of programs \eqref{eq:sdp} and \eqref{eq:dual} are equal. Furthermore, if ($\mu^*, \lambda^*$) is an optimal solution of \eqref{eq:dual}, setting
\[A = \lambda^* N + \mu^* I - M,\]
we have that a feasible solution $\rho$ of $\eqref{eq:sdp}$ is optimal if and only if complementary slackness is satisfied, i.e., if and only if $A\rho = 0$.

Let $\ket{\psi^*}$ be a pure optimal solution of \eqref{eq:sdp}, whose existence is guaranteed by Lemma \ref{lem:pure}. $\ket{\psi^*}$ can be assumed to lie in the space spanned by the ranges of $M$ and $N$ without loss of generality. Let $v_1, \ldots, v_{d_1}$ and $w_1, \ldots, w_{d_2}$ be Jordan bases of the range of $M$ and $N$, respectively, as promised by Lemma \ref{lem:jordan}. For $i \leq \min(d_1, d_2)$, let $T_i$ be the subspace spanned by $v_i$ and $w_i$, and for $\min(d_1, d_2) < i \leq \max(d_1, d_2)$, let $T_i$ be the subspace spanned by $v_i$ if $d_1 > d_2$ and by $w_i$ otherwise. We can then write
\[\ket{\psi^*}= \sum_i \sqrt{p_i} \ket{\psi_i}\]
with $\ket{\psi_i} \in T_i$ for every $i$. Since $\ket{\psi^*}$ is a primal optimal solution, it follows that
\[0 = \braket{\psi|A|\psi} = \sum_{i,j} \sqrt{p_i p_j} \braket{\psi_i|A|\psi_j} = \sum_i p_i \braket{\psi_i|A|\psi_i}.\]
The last equality above follows because the $T_i$ are mutually orthogonal and invariant under multiplication by $M$ or $N$, and hence by $A$. Because $A$ is positive semidefinite, this implies that $\ket{\psi_i}$ satisfies complementary slackness for every $i$ with $p_i > 0$. Also, similarly, 
\[t \leq \braket{\psi | N | \psi} = \sum_i p_i \braket{\psi_i | N | \psi_i}.\]
This shows that there exists some $i$ such that $\ket{\psi_i}$ is an optimal solution of the primal program. It can be seen that we can assume without loss of generality that $i \leq \min(d_1, d_2)$, as otherwise either the objective is zero or the constraint is violated. It follows that there exists an optimal primal solution of program \eqref{eq:sdp} which is sitting in one of the Jordan blocks $T_i$ of dimension $2$. Hence, in view of Lemma \ref{prop:optsdpspecial}, the optimal value of the program \eqref{eq:sdp} is:
\begin{align*}
\max\limits_{i \in [\min(d_1, d_2)]} F(\left|\braket{v_i|w_i}\right|^2, t) &= F\left(\max\limits_{i \in [\min(d_1, d_2)]} \left|\braket{v_i|w_i}\right|^2, t\right)\\
                        &= F\left(\|MNM\|_\infty, t\right)
\end{align*}
where the fact that $F$ is monotonic in the first argument was used for the first equality. This completes the proof.
\end{proof}

\subsection{Finishing the proof}\label{sec:finproof}
With the result of the previous section in hand, we can prove Theorem \ref{thm:eigvalbound}, thereby finishing the proof of Theorem \ref{thm:minentropy}. We start by proving the following concentration inequality.
\begin{lemma} \label{lem:concentration}
    Let $Q$ be a projector on $B_1B_2$. Define the function $f: G(d,k) \rightarrow \mathbb{R}$ by:
    \[f(P) = \|Q(I^{B_1} \otimes P^{B_2})Q\|_\infty.\]
    If $P$ is sampled from the Haar measure on $G(d,k)$, then, for some universal constant $c' > 0$, we have that for all $t > 0$
    \[\Pr\left[|f(P) - \mathbb{E}[f(P)]| \geq t\right] \leq \exp(-c'dt^2)\]
    
\end{lemma}
\begin{proof}
Given $P_1, P_2$, we have
\begin{align*}
|f(P_1) - f(P_2)| &\leq \|Q(I \otimes (P_1-P_2))Q\|_\infty\\
              &\leq \|I \otimes (P_1-P_2)\|_\infty\\
              &= \|P_1-P_2\|_\infty\\
              &\leq \|P_1-P_2\|_2
\end{align*}
where the reverse triangle inequality was applied for the first inequality. The result then follows from the concentration inequality for Lipschitz functions in the Schatten $2$-norm in Fact \ref{prop:concentration}.                            
\end{proof}

We will need to calculate the expectation of $f(P)$ in order to use the concentration inequality above. We do this in the following lemma.
\begin{lemma} \label{lem:expectation}
Let $Q$ be a rank-$l$ projector on $B_1B_2$. Defining $f: G(d,k) \rightarrow \mathbb{R}$ by:
\[f(P) = \|Q(I^{B_1} \otimes P^{B_2})Q\|_\infty\]
We have that
\[\mathbb{E}_{P \sim \mu_{\text{Haar}}}[f(P)] \leq \frac{k}{d} + O\left(\sqrt{\frac{l}{d}}\right).\]
\end{lemma}
\begin{proof}
We begin by noting that, using the triangle inequality,
\begin{align*}
\mathbb{E}[\|Q(I \otimes P)Q\|_\infty] &\leq \frac{k}{d} \mathbb{E}[\|Q(I \otimes I)Q\|_\infty]  + \mathbb{E}\left[\left\|Q\left(I \otimes \left(P - \frac{k}{d}I\right)\right)Q\right\|_\infty\right] \\
                        &= \frac{k}{d} + \mathbb{E}\left[\left\|Q\left(I \otimes \left(P - \frac{k}{d}I\right)\right)Q\right\|_\infty\right].
\end{align*}
All that remains to do is to bound the latter term, which is a routine exercise in random matrix theory. Take $\ket{v_1}, \ldots, \ket{v_N}$ to be a $(1/2)$-net over the complex sphere in the range of $Q$. This range is assumed to be of dimension $l$, so that we may take $N = C^l$ as per Fact \ref{fc:sphere-net}, for some constant $C > 1$. For $i \in [N]$, define the random variable $X_i$ by
\[X_i = \left\langle v_i \middle| \left(I \otimes \left(P - \frac{k}{d}I\right)\right) \middle| v_i\right\rangle.\]
Setting $X = \max_i |X_i|$, Fact \ref{lem:net2} gives that for any $P$, 
\[\left\|Q\left(I \otimes \left(P - \frac{k}{d}I\right)\right)Q\right\|_\infty \leq 2 X.\]
Because $\mathbb{E}[P] = \frac{k}{d}I$, it holds that $\mathbb{E}[X_i] = 0$ for all $i$. We can do calculations similar to those in Lemma \ref{lem:concentration} to show that $X_i$ is $1$-Lipschitz in $P$. Then applying Fact \ref{prop:concentration} we have that, for any $t$,
\[ \Pr[X_i > t] \leq \exp(-c'dt^2)\]
for all $i$. Then by a union bound we get, 
\[\Pr[X > t] \leq N \exp(-c'dt^2).\]
For any $t_0 > 0$, we can upper bound $\mathbb{E}[X]$ as follows:
\begin{align*}
\mathbb{E}[X] &= \int_0^\infty \Pr[X > t] dt \\
     &= \int_0^{t_0} \Pr[X > t] dt + \int_{t_0}^\infty \Pr[X > t] dt \\
     &\leq t_0 + \frac{N}{2c't_0 d} \exp(-c'dt_0^2)
\end{align*}
where the bound in Fact \ref{lem:gaussianbound} was used in the last inequality. Setting $t_0 = \Theta\left(\sqrt{\frac{l}{d}}\right)$ makes the second term in $O\left(\frac{1}{\sqrt d}\right)$ and hence proves the result.
\end{proof}

We can now complete the proof of the main theorem of this section.
\begin{proof}[Proof of Theorem~\ref{thm:eigvalbound}]
Let $G$ and $\omega$ and $\nu$ be as in the statement of the theorem. A well-known characterization of the sum of the largest $l$ eigenvalues is
\[ \sum_{j=1}^l\nu_j = \sup_{0 \preceq Q \preceq I, \text{rank}(Q)=l}\Tr(Q\omega).\]
Moreover, the optimal $Q$ may be taken to be a projector. Therefore, it is enough to bound $\Tr(Q\omega)$ for any rank-$l$ projector $Q$. 
Let $Q$ be such a projector. For $P \in G(d,k)$, let $f(P$) be as defined in Lemma \ref{lem:expectation}, 
and set $t(P) = \Tr((I^{B_1} \otimes P^{B_2}) \omega_P)$. By definition, we have, $\mathbb{E}_{P \sim \mu}t(P) = 1-\eps$. Lemma \ref{prop:optsdp} then gives:
\begin{align*}
\Tr(Q\omega) & = \mathbb{E}_{P\sim \mu}[\Tr(Q\omega_P)]\\
            &\leq \mathbb{E}_{P \sim \mu}[F(f(P), t(P))] \\ 
           &\leq F(\mathbb{E}_{P \sim \mu} [f(P)], 1 - \varepsilon)
\end{align*}
where the concavity of $F$ was used for the second inequality.  It is therefore enough to control the expectation of $f(P)$ under  $P\sim \mu$. We have already calculated its expectation under the Haar distribution in Lemma \ref{lem:expectation}. Using Fact \ref{lem:expectationbound} we therefore get for $r=\log K$, 
\begin{align*}
\mathbb{E}_{P \sim \mu} [f(P)] &\leq \mathbb{E}_{P \sim \mu_\text{Haar}} [f(P)] + \mathbb{E}_{P \sim \mu} \left[\left|f(P) - \mathbb{E}_{P \sim \mu_\text{Haar}} [f(P)]\right|\right] \\
                            &\leq \frac{k}{d} + O\left(\sqrt{\frac{l}{d}}\right) + 2\left(\mathbb{E}_{P \sim \mu_\text{Haar}} \left[\left|f(P) - \mathbb{E}_{P \sim \mu_\text{Haar}} [f(P)]\right|^r\right]\right)^{1/r}.
\end{align*}

Finally, we can upper bound the last term in the above expression using Fact \ref{prop:subgaussian} along with Lemma \ref{lem:concentration}. This gives us,
\[ \left(\mathbb{E}_{P \sim \mu_\text{Haar}} \left[\left|f(P) - \mathbb{E}_{P \sim \mu_\text{Haar}} [f(P)]\right|^r\right]\right)^{1/r} \leq O\left(\sqrt{\frac{r}{d}}\right) \leq O\left(\sqrt{\frac{\log K}{d}}\right).\]
Putting everything together, we get,
\[ \sum_{j=1}^l\nu_j = \Tr(Q\omega) \leq F\left(\frac{k}{d} + O\left(\sqrt{\frac{l}{d}} + \sqrt{\frac{\log K}{d}}\right), 1-\eps\right). \qedhere \]
\end{proof}

\section{Applications}
\subsection{Incompressibility of an ensemble of flat states}\label{sec:incompress}
In this section, we prove our result about the impossibility of visible compression of an ensemble of flat states, as a corollary of our entanglement lower bound for RSP in Theorem \ref{thm:minentropy}. First we state the definition of visible compression (without entanglement).
\begin{definition}[Visible state compression without entanglement]
Fix an ensemble $\mathcal{E} = (\mu, \rho)$ of quantum states in $\mathbb{C}^d$. A visible $(d',\eps)$-compression scheme for $\mathcal{E}$ is an encoding map $E: \mathbb{C}^d \to \mathbb{C}^{d'}$ and a decoding map $D: \mathbb{C}^{d'}\to\mathbb{C}^d$ such that $D$ is CPTP ($E$ can be any measurable map), and $d'\leq d$, such that
\[ \mathbb{E}_{\rho \sim \mu}\left[\|D\circ E(\rho) - \rho\|_{\tr}\right] \leq \eps.\]
\end{definition}
We now prove Theorem \ref{cor:incompress}, whose statement is recalled below.
\incompress*
\begin{proof}
Such a compression procedure gives rise to a $(d,k)$-RSP protocol which consists in Alice preparing $E(\rho)$ on her end and teleporting it to Bob, who then applies the decompression map $D$ to get $D\circ E(\rho)$, and this is close to $\rho$. This protocol requires $\log d'$ EPR pairs and $2\log d'$ bits of communication, and the average error of the protocol is the same as that of the original procedure. Take $\eta > 0$, and take $d_0$ to be the smallest value of $d$ such that if the second term in the first argument of the fidelity function in the statement of Theorem \ref{thm:minentropy} is $A\sqrt{\frac{m}{d}}$, then
\[A\sqrt{\frac{2\log d}{d}} < \frac{\eta}{2}.\]
Setting $\delta = F\left(\frac{k}{d}+\frac{\eta}{2}, 1-\varepsilon\right)$, we see that $\delta < 1$ for any $\eps$ as in the statement of the corollary. Moreover, setting $\gamma = \frac{1-\delta}{2}$, and noting that $F$ is monotonically increasing in the first argument, we can compute $H^{\delta+\gamma}$ of $\log d'$ EPR pairs:
\[H^{\delta + \gamma}_{\text{min}}(A)_\sigma = \log d' - \log(1-\gamma - \delta) = \log d' - \log(1 - \delta) + 1. \]
Now Theorem \ref{thm:minentropy} implies that
\[ \log d' - \log(1-\delta) + 1 \geq \log d + 3 \log(1-\delta) - 3,\]
which can be rearranged to give the result.
\end{proof}

\subsection{An entanglement-optimal bounded-error protocol for the equality function}
We believe that the protocols we gave for $(d,k)$-RSP could turn out to be a useful subroutine for designing entanglement-assisted protocols in the future. Here we describe such a protocol for computing the equality function on $n$ bits.

In the $\mathrm{EQ}_n$ problem, Alice and Bob are given $x,y \in \{0,1\}^n$ and need to determine whether they are equal. It is known that $n$ bits (or $n$ qubits in the absence of shared entanglement) of communication are necessary and sufficient to achieve this with probability $1$, but allowing for a small error probability $\varepsilon > 0$ changes the story significantly. It is known that, classically, the sharing of $\log n + O(\log \frac{1}{\eps})$ public random bits, along with communication $\log \frac{1}{\varepsilon} + O(1)$ (which crucially does not depend on $n$) are necessary and sufficient. Clearly, the same is possible if Alice and Bob share this many EPR pairs instead of random bits as they can be converted into random bits by measuring them in the computational basis.

Our protocol for $\mathrm{EQ}_n$ will instead use $O(1)$ communication and only $\frac{1}{2}\log n$ many EPR pairs. Our result is formally restated below.
\eq*
The protocol achieving Theorem \ref{thm:equpper} is given in Figure \ref{fig:proteq}.

\begin{center}
\begin{mdframed}[nobreak=true]
\textbf{Given.}
\begin{itemize}
\item A collection $\{P_x\}_{x \in \{0,1\}^n}$ of elements of $G(d,k)$
\item An error parameter $\eps > 0$
\end{itemize}

\textbf{Input.}
Alice and Bob get inputs $x \in \{0,1\}^n$ and $y\in \{0,1\}^n$ respectively.

\textbf{Protocol.}
\begin{enumerate}
  \item Alice and Bob run a $(d,k)$-RSP protocol with worst-case error $\frac{\eps}{2}$ given Alice input $P_x$. 
  \item Bob does the $\{P_y, I-P_y\}$ measurement on the target register of the RSP protocol; he declares that $x=y$ if he obtained the first outcome and $x\neq y$ otherwise.
\end{enumerate}

\end{mdframed}
\captionsetup{hypcap=false}
\captionof{figure}{An entanglement-assisted protocol for the equality function on $n$ bits}
\label{fig:proteq}
\end{center}

It is easy to see that the worst-case error of the protocol in Figure \ref{fig:proteq} is at most
\[\frac{\eps}{2} + \max_{x \neq y} \frac{1}{k} \Tr(P_x P_y).\]
Therefore, we need to pick $\{P_x\}_x$ such that $\max_{x\neq y} \Tr(P_xP_y) \leq \frac{\eps}{2}$ in order for the protocol to succeed with probability $1-\eps$. 
The following lemma shows that for $d = \Theta\left(\frac{\sqrt{n}}{\varepsilon^{3/2}}\right)$ and $k = \Theta(\eps d)$, such a choice of $P_x$'s exist. This proves Theorem \ref{thm:equpper}, assuming we use our Kraus operator protocol for $(d,k)$-RSP.

\begin{lemma}
    Take $d \in \mathbb{N}$, and for $\varepsilon \in (0,1)$, set $k = \lceil\frac{\varepsilon d}{2} \rceil$. For $m=\Theta(d^2 \varepsilon^3)$, there exist projectors $P_1, \ldots, P_{2^m} \in G(d,k)$ such that, for all $i \neq j$, we have
    \[\frac{1}{k} \Tr(P_iP_j) < \frac{\varepsilon}{2}.\]
\end{lemma}

\begin{proof}
As in the rest of the paper, the existence of these projectors is shown using the probabilistic method. Note that for a fixed projector $P_0$, the function $f: G(d,k) \mapsto \mathbb{R}$ given by $f(Q) = \frac{1}{k} \Tr(QP_0)$ has Lipschitz constant $\frac{1}{\sqrt{k}}$. Therefore, in view of the concentration inequality in Fact \ref{lem:concentration}, if we sample $P_1, \ldots,  P_N$ uniformly at random in $G(d,k)$, the probability that there exist $i \neq j$ with $\frac{1}{k} \Tr(P_iP_j) >\frac{\varepsilon}{2}$ is bounded by
\[ \binom{2^m}{2} \cdot \exp\left(-\frac{c'}{4}dk\varepsilon^2\right).\]
This is smaller than $1$ provided that 
\[2^{2m}  \leq \exp\left(\frac{c'}{8}dk\varepsilon^2\right) \leq \exp\left(\frac{c'}{8}d^2\varepsilon^3\right)\]
which is satisfied for the value of $m$ given in the statement of the lemma.
\end{proof}



\section*{Acknowledgments}
We thank Ashwin Nayak for useful conversations about the paper \cite{Bab_Hadiashar_2020}, as well as Debbie Leung for useful conversations about the paper \cite{bennett2005remote}. We thank Rahul Jain for pointing out the Brothers Extension method in \cite{AJM+16} to us as well as for pointing out a mistake in an earlier version of the paper. Although all the high-level ideas in this paper were human-generated and everything in the paper was written by the two authors, we acknowledge the use of various LLMs such as GPT-5, Gemini and Claude to pin down certain technical details. The second author is supported by NSERC.

\printbibliography

\end{document}